\documentclass[journal]{IEEEtran}

\def\figpath{./}

\newif\ifjournal
\journaltrue

\newif\ifappendixneeded
\newif\ifuncertain

\usepackage{multirow}
\usepackage{graphicx}

\usepackage{url}

\usepackage{tikz}

\usepackage{bm}
\usepackage{amsmath} 

\usepackage{mathrsfs}

\usepackage{algorithm}
\usepackage{algorithmic}

\usepackage{color}

\usepackage{amsthm}
\theoremstyle{definition}
\newtheorem{lemma}{Lemma}
\newtheorem{theorem}{Theorem}

\newtheorem{problem}{Problem}
\newtheorem{definition}{Definition}

\usepackage{amssymb}

\newcommand\ceiling[1]{\left\lceil #1 \right\rceil}

\title{DeepPR: Progressive Recovery for Interdependent VNFs with Deep Reinforcement Learning}

\author{Genya Ishigaki, \IEEEmembership{Student Member, IEEE,} Siddartha Devic, Riti Gour, \IEEEmembership{Student Member, IEEE,} and~Jason~P.~Jue,~\IEEEmembership{Senior~Member,~IEEE}
\thanks{Manuscript submitted \today.} 
\thanks{Genya Ishigaki, Siddartha Devic, Riti Gour, and Jason P. Jue are with the Department of Computer Science at The University of Texas at Dallas, Richardson Texas 75080, USA (Email: \{gishigaki, sid.devic, rgour, jjue\}@utdallas.edu).

An earlier version of this paper \cite{my_globecom} will be presented at IEEE Global Communications Conference (GLOBECOM) 2019.
}
  
  }

\begin{document}

\maketitle
\IEEEpeerreviewmaketitle

\begin{abstract}
The increasing reliance upon cloud services entails more flexible networks that are realized by virtualized network equipment and functions. When such advanced network systems face a massive failure by natural disasters or attacks, the recovery of the entire system may be conducted in a progressive way due to limited repair resources.
The prioritization of network equipment in the recovery phase influences the interim computation and communication capability of systems, since the systems are operated under partial functionality.
Hence, finding the best recovery order is a critical problem, which is further complicated by virtualization due to dependency among network nodes and layers.
This paper deals with a progressive recovery problem under limited resources in networks with VNFs, where some dependent network layers exist. We prove the NP-hardness of the progressive recovery problem and approach the optimum solution by introducing DeepPR, a progressive recovery technique based on Deep Reinforcement Learning (Deep RL). 
Our simulation results indicate that DeepPR can achieve the near-optimal solutions in certain networks and is more robust to adversarial failures, compared to a baseline heuristic algorithm. 

\end{abstract}
\begin{IEEEkeywords}
resource allocation, Deep Reinforcement Learning (Deep RL), network recovery, Network Function Virtualization (NFV), interdependent networks.
\end{IEEEkeywords}

\section{Introduction}


Resilience is a critical concern for communication networks that are deployed in support of cloud systems. However, the recent trend towards the virtualization of network equipment and functions potentially introduces new fragility into such systems due to layering \cite{coNEXT18_failure}.


Many studies reveal the fragility that is unique in layered networks \cite{7367911,7498079}. For example, a network system may be realized by the combination of virtualized functions and infrastructure (physical) nodes. The nodes in the infrastructure layer host some functions including the orchestrator function that manages the life cycle of virtualized functions and the mapping between the two layers \cite{ETSIstandard}. The functionality of the orchestrator function depends on the infrastructure node hosting it; at the same time, it is necessary for an infrastructure node to be reachable to a working virtualized orchestrator that manages the physical computation resources on every infrastructure node. This interdependency between two layers results in increased fragility.



Furthermore, the interdependency has an influence on recovery decisions after a massive failure. After massive failures, it is critical to start providing necessary connections and services as soon as possible, even when available resources, such as manpower or backup equipment, to repair the system is limited. The prioritization of specific connections or services is well-studied in \cite{5934996,7440669,8057042} for single layer networks. However, this prioritization becomes more complex when there is interdependency between layers, since the role of each node is determined not only by the topology of a network but also by the interdependency \cite{7842042,NaturecommMajdandzic}. The following example characterizes the inherent complexity of the problem.

Let us consider an example illustrated in Figure \ref{motivating}. The network consists of two constituent layers, which represent a virtualized function layer $G_0$ and an infrastructure layer $G_1$. Each server $v_i$ on $G_1$ can host one function $f_i$. Suppose that either $v_1$ or $v_2$ hosts a virtualized orchestration function among the four servers; i.e. $f_1$ or $f_2$ can be an orchestration function. As explained above, at least one orchestration function needs to be available for servers to be functional. The demand of each server shows the amount of resources needed to repair it.

\def\scale{0.48}
\begin{figure}[t]
\centering
\includegraphics[width=\scale\textwidth]{\figpath 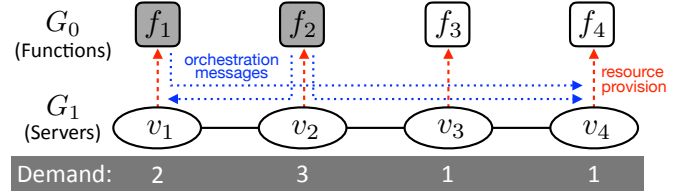}
\caption{A Motivating Example: Every physical node requires a connection to at least one control node in the function layer to receive orchestra-ion messages. Also, each function node needs to be provided with computation resources by the physical node hosting it.}
\label{motivating}
\end{figure}
\begin{table}[t]
\centering
\caption{Available Computation Power at Each Time Step: The difference in the recovery order causes loss of potential cumulative interim computation power.}
{\renewcommand\arraystretch{1.5}
\begin{tabular}{c||c|c|c|c|c|c|c}
\hline
& $t_1$ & $t_2$ & $t_3$ & $t_4$ & $t_5$ & $t_6$ & $t_7$\\
\hline
$P_1$&0&1 $(v_1)$&1 & 1 &2 $(v_2)$ & 3 $(v_3)$ & 4 $(v_4)$\\
$P_2$&0&0&0&0&3 $(v_2, v_3, v_4)$& 3 & 4 $(v_1)$\\
\hline
\end{tabular}
}
\label{samplerecv}
\end{table}

Our problem is to determine the recovery order of the servers, considering the number of functions available during the recovery process. Here, the following two recovery orders are compared in terms of the total number of functions available over recovery time steps: $P_1: v_1 \rightarrow v_2 \rightarrow v_3 \rightarrow v_4$ and $P_2: v_4 \rightarrow v_3 \rightarrow v_2 \rightarrow v_1$. For simplicity, it is assumed that only one unit of resource is available at each time step $t_i\ (1 \leq i \leq 7)$.

Table \ref{samplerecv} describes the number of available functions at each time step when following each recovery order. Note that an integer in each cell represents the the number of functions available (utility) at the time step. For instance, in $P_1$, we first recover $v_1$ and obtain 1 available function (utility) at $t_2$, since it takes two steps to satisfy the demand of the node. A recovered node stays functional until the last step $t_7$ and continues providing the same utility at every step after the step in which it was recovered. Therefore, the computation capability at $t_3$ and $t_4$ is 1, as there are no other nodes recovered during these steps. Since $v_2$ is recovered after three steps, another unit of utility is added at $t_5$. In $P_2$, the interdependency between the virtualized function layer and the infrastructure layer plays an interesting role in the recovery process. Even though sufficient resources are assigned to $v_4$ and $v_3$ in the first two steps, the utility remains 0 until $v_2$ is recovered. This is because the two nodes ($v_3, v_4$) cannot receive the orchestration messages due to the unreachability to $f_2$, which is an orchestration function. Hence, the total utility jumps to 3, once $v_2$ is recovered at $t_5$. As a result, the total utility over time of $P_1$ is 12, while the total utility of $P_2$ is 10.

Hence, the total utility available during recovery is different depending on which recovery order we adopt. Motivated by this simple example, the question addressed in this paper is the following. \textit{How do we find a recovery order that maximizes the accumulated utility during the recovery process in networks with interdependency between layers?} This problem is a variant of the progressive recovery problem \cite{5934996}, which aims at maximizing the amount of flows going through a network during the recovery process. However, the fundamental difference lies in the consideration of the interconnectedness between nodes in different layers.


Our major contribution is twofold: a set of theoretical results, which narrow down decision-making factors in the recovery, and a Deep Reinforcement Learning-based (Deep RL) algorithm to decide the recovery order. Being combined with the theoretical results that provide guidelines on the selection of a recovery order, the Deep RL technique demonstrates its performance as a general method to solve the recovery problem, which answers the research question above.
The following are the key contributions of the rest of this paper.
\begin{itemize}
\item To the best of our knowledge, this is the first paper defines the progressive recovery problem in networks with layer dependency.
\item We prove that the progressive recovery problem with a general graph always has an equivalent progressive recovery problem with a simpler graph (Section \ref{problem_conversion} - Theorem \ref{general_are_onelayer}).
\item The NP-hardness of the simpler problem is shown, which implies that the general case cannot be solved in polynomial time (Section \ref{problem_section} - Theorem \ref{general_star_NP}).
\item A heuristic algorithm (RATIO) is proposed, and its limitation is described by introducing an adversarial scenario (Section \ref{ratio_section}).
\item A Deep reinforcement learning-based algorithm for Progressive Recovery (DeepPR) is proposed, integrating the RATIO heuristic, to deal with the limitation of the heuristic (Section \ref{DeepPR_section}). 
\item Our simulation results indicate that DeepPR could achieve near-optimal solutions and is robust against the adversity, using the exploration with RATIO.
\item Our results suggest that the integration of reinforcement learning and a heuristic algorithm that is specifically designed for an optimization problem provides a mean to solve optimization problems more effectively than simple use of reinforcement learning or heuristic algorithms (Section \ref{disc_section}).
\end{itemize}

Note that an earlier version of this paper will be presented at IEEE Global Communications Conference (GLOBECOM) 2019 \cite{my_globecom}. This journal paper provides a complete view of the theoretical results partially discussed in the conference version, and several new discussions and evaluations related to adversarial failure scenarios.


\section{Related Works}
Pioneering work \cite{5934996} on the progressive recovery problem focuses on determining the recovery order of communication links that maximizes the amount of flows on the recovered network with limited resources. As an extension, the work \cite{7440669} proposes node evaluation indices to decide the recovery order to maximize the number of virtual networks accommodated. Considering the necessity of monitoring to observe failure situations, the joint problem of progressive recovery and monitor placement is discussed in \cite{8057042}.

%

The fragility induced by dependency between network layers has been pointed out in the context of interdependent network research \cite{7367911,7498079,generalsurvey_cascading}. In particular, the interdependency between virtualized nodes and physical nodes in optical networks is considered in \cite{7367911}. A similar dependency caused by VNF orchestration is discussed in \cite{7498079}.



The works in \cite{STIPPINGER2014481, PhysRevE.92.052806, 7913673} analyze the behaviors of failure propagations in such interdependent networks when each node performs local recovery (healing), where a functioning node substitutes for the failed node by establishing new connections with its neighbors. 


Progressive recovery problems in interdependent networks have been discussed in \cite{interdepprogrecov_ealier,7842042, NaturecommMajdandzic,4343992}. Classifying the progressive recovery problems by the types of interdependency, the work \cite{interdepprogrecov_ealier} proposes the optimum algorithm for a special case and heuristic algorithms for other cases. ILP and Dynamic Programming-based algorithms are employed to solve a variant of the progressive recovery problem in \cite{7842042}. 

\ifjournal
Other works \cite{6849338,6407159} propose some metrics to evaluate network nodes that can be used to decide the priority among the nodes.
\fi



%



\section{Model}
\subsection{Network Model}
A network, which consists of virtulized functions and infrastructure nodes hosting the functions, is modeled by an interdependent network that is formed by two constituent graphs $G_i = (V_i, E_{ii})\ (i \in \{0, 1\})$, which correspond to the virtualized orchestration function layer $(G_0)$ and the infrastructure node layer $(G_1)$. A pair of nodes in different constituent graphs can be connected by an arc representing their dependency relationships: $A_{ij}\ (i, j \in \{0, 1\},\ i \neq j)$. Edges in $E_{ii} \subseteq V_i \times V_i$ are called \textit{intra-}edges because they connect pairs of nodes in a constituent network. In contrast, arcs in $A_{ij} \subseteq V_i \times V_j\ (i \neq j)$ are called \textit{inter-} or \textit{dependency} arcs. An arc $(v_i, v_j) \in A_{ij}\ (v_i \in V_i, v_j \in V_j)$ indicates that a node $v_j$ has dependency on a node $v_i$. The node $v_i$ is called a \textit{supporting} node, and $v_j$ is a supported node. 


Two node attribute functions are defined to capture the characteristics of each node: \textit{demand} and \textit{utility} functions. The demand function $d: V \rightarrow \mathbb{N}$ represents how many resources needs to be assigned to fully recover a given node. This demand can be interpreted as the cost or manpower to repair a specific node in the context of recovery problems. The utility function $u: V \rightarrow \mathbb{N}$ indicates the computational capability of a given node, such as the number of functions it can host, when it is fully recovered.

\subsection{Network Failure and Progressive Recovery Plan}
When a network failure event occurrs at time $t_0$, some nodes in the network become nonfunctional. Let $F[t_k] \subseteq V (= \bigcup_{i \in \{0,1\}} V_i)$ denote a set of nonfunctional nodes at time $t_k$. With this notation, the nonfunctional nodes right after the failure are represented as $F[t_0]$. A failure is represented by a node set in this paper, because any failure of an edge can be converted to a node failure by replacing the nonfunctional edge $(v_i, v_j) \in E$ with a nonfunctional node $v_{ij}$ and two functional edges $\{(v_i, v_{ij}), (v_{ij}, v_j)\}$. 

In progressive recovery scenarios, we receive a limited amount of resources at each time step after a failure. The resource function $r: t_i \mapsto c_i \in \mathbb{N}$ indicates the amount of the repair resources available at time $t_i\ (i \in \{0, ..., T\} \subset \mathbb{N})$.

A progressive recovery plan $P$ is an assignment of the available resources to the nonfunctional nodes. Formally, $P$ is a $(T+1) \times |V|$ matrix whose entries indicate the amount of resources assigned to a specific node at a specific time. Because of the limitation on the available resource amount, $P[t_i] \left(:= \sum_{v \in V} P[t_i][v]\right) = r(t_i)$ for every $t_i$.

\def\Satu{K}

During the recovery process, nodes can be classified by two measures: the amount of resources assigned to the node and the functionality of the node. A node $v$ is \textit{saturated} when it has received enough recovery resources: $d(v) \leq \sum_{i = 0}^{k} P[t_i][v]$. Let $\Satu[t_i]$ denote a set of saturated nodes at time $t_i$. A node $v$ is said to be \textit{functional} if and only if it is (1) saturated and (2) reachable from at least one saturated supporting node in the other constituent graph via a simple path consisting of functional nodes. When a node $v$ is functional at time $t_i \ (i \in \{0, ..., T\} \subset \mathbb{N})$, the node state function $\alpha_i(v) = 1$; otherwise $0$. A node $v$ is \textit{recovered} at $t_i$ only when it becomes functional by assigning $P[t_i][v]$. In real networks, a nonfunctional saturated node can be interpreted as either an infrastructure node unreachable from an orchestration function or a virtualized function that is hosted on an infrastructure node that is nonfunctional.

A resource assignment $P[t_i]$ at each step $t_i$ is called a \textit{splitting} assignment when it prevents any nodes from saturation or recovery, even though there exists a node that can be saturated or recovered at $t_i$. Contrarily, a \textit{concentrating} assignment saturates or recovers some node if possible, and provides all the extra resources, which cannot saturate nor recover any node, to one unsaturated node.

\ifuncertain
\subsection{Certainty vs Uncertainty Model}
In a certainty model, we assume the access to the states $\alpha$ of all nodes. This means we always know which nodes are functional or nonfunctional through an imaginary network-wide view.

However, the monitoring capability for the network-wide view may be too optimistic in case of massive system failures. Therefore, an uncertainty model restricts the access to node states $\alpha$ to local view. More precisely, we can know whether or not a node is functional only when the node is adjacent to functional nodes. This will be interpreted as the monitoring function attached to each node becomes functional, when the associated node is fully recovered.
\fi


\def\prob{\mathbb{P}}
\section{Problem Formulation \label{problem_section}}
This section formulates the progressive recovery problem in interdependent networks, and discusses and proves some properties of the problem. 

\subsection{The Problem and Special Cases}
The progressive recovery problem is to find a recovery plan $P$ represented by a (time step $\times$ node)-matrix that maximizes the sum of utility provided by functional nodes during the recovery.

\begin{problem} \textbf{Progressive Recovery Problem (PR)}: 
Given a graph $N = (V = V_0 \cup V_1, A = A_{01} \cup A_{10} \cup E_{00} \cup E_{11})$, a demand function $d$, an utility function $u$, a set of initially failed nodes $F[t_0] \subseteq V$, and a resource function $r$, maximize the network-wide utility $U_P = \sum_{i = 0}^{T} \sum_{v \in V} u(v) \alpha_i(v)$ by deciding a resource assignment matrix $P$.
\end{problem}



\def\scale{0.9}
\begin{figure*}[t]
\centering
\begin{minipage}{0.47\textwidth}
\centering
\includegraphics[width=\scale\columnwidth]{\figpath 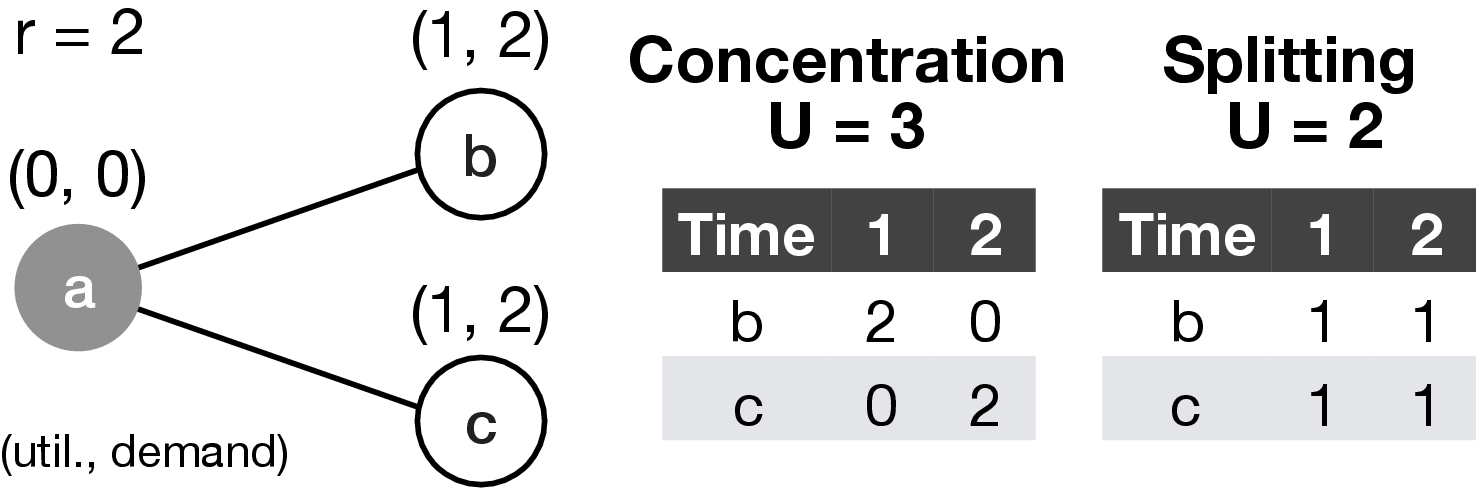}
\caption{Concentration vs. Splitting: Available resources at a time step should be concentrated to a set of nodes as much as possible.\label{concent_fig}}
\end{minipage}\hfill
\begin{minipage}{0.47\textwidth}
\centering
\includegraphics[width=\scale\columnwidth]{\figpath 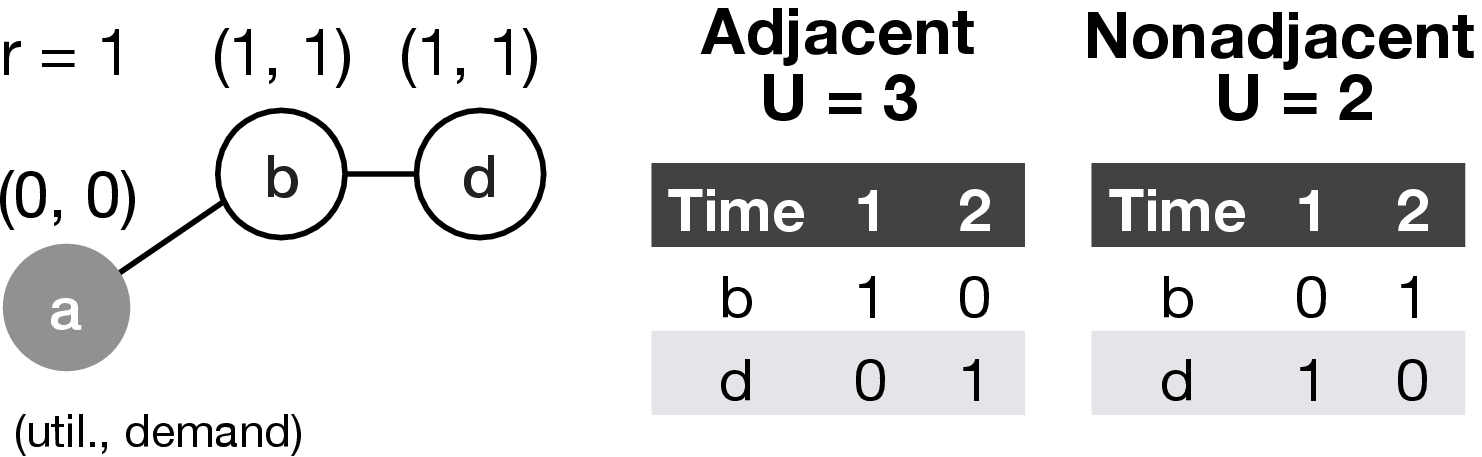}
\caption{Allocation and the Adjacency to Functional Nodes: The nodes closer to functional nodes should be prioritized in a recovery order.\label{adj_fig}}
\end{minipage}
\end{figure*}

\ifuncertain
\begin{problem} (Uncertainty)
Given a graph $N = (V = V_0 \cup V_1, A = A_{01} \cup A_{10} \cup E_{00} \cup E_{11})$, a demand function $d$, an utility function $u$, a set of visible failed nodes $F[t_0] \subseteq F$, and a resource function $r$, maximize the network-wide utility $U_P = \sum_{i = 0}^{T} \sum_{v \in V} u(v) \alpha(v)$ by deciding a resource assignment matrix $P$, when $F(t_i)\ (i \in \{0, ..., T\})$ denotes the set of nodes in $F$ that are adjacent to the functional nodes at time $t_i$.
\end{problem}
\fi



A simpler case of the problem is one in which it is assumed that the functionality of virtualized functions totally depends on the functionality of a physical server hosting the function. In other words, there is no need for the assignment of recovery resources to repair virtualized functions, since the unavailability of the functions occurs only due to the loss of physical servers hosting them. In our terminology, when virtualized function nodes are nonfunctional, they are always saturated. 

The interdependency between the virtual and physical layer still exists even with the above assumption, since any physical machine needs at least an indirect connection with a virtual control function. Obviously, a virtual function needs at least one physical machine, which can host it, to be functional. 

\begin{definition}
A graph $N = (V, A)$ in the progressive recovery problem is said to be \textit{one-layered} when nodes in $G_0 = (V_0, E_{00})$ never require repair resources to be functional. In other words, nodes in $G_0$ are nonfunctional only because the loss of supporting nodes in the other constituent graph: $v \in \Satu[t_0]$ for any node $v \in (V_0 \cap F[t_0]$). 
\label{def-one-layer}
\end{definition}

\begin{problem}
\textbf{One-layered star case (StarPR)}: 
\label{StarProb}
Assume that the graph topology is a star whose nodes are in $G_1$, except for the center node $v \in V_0$; also, each node $u \in V_1$ is biconnected with $v \in V_0$.
\end{problem}

\begin{problem}
\label{RTreeProb}
\textbf{One-layered rooted tree case}: Extend Problem \ref{StarProb} by adding more nodes to $G_1$ that are not adjacent to the node in $G_0$. i.e., the graph is a tree rooted at the node $v$ in $G_0$.
\end{problem}



\subsection{Intractability}

\ifjournal
\begin{definition} \textbf{Time-Invariant Incremental Knapsack Problem (IIK)} \cite{incremental_knapsack}: Let $X = \{x_i\}$ denote a set of items, which each have value $a(x_i)$ and weight $w(x_i)$. For any subset $X'$ of $X$, the value and weight are defined as follows: $a(S) = \sum_{x_i \in S} a(x_i)$, and $w(X') = \sum_{x_i \in S} w(x_i)$. IIK is to find a sequence of subsets of $X$, $[S_1, S_2, ..., S_T]\ (S_i \subseteq S_{i+1},\ i = 1, ..., T-1)$ from time 1 to $T$ that maximize $\sum_{t=1}^T a(S_t)$ subject to $w(S_t) \leq B_t\, (t = 1, ..., T)$, where $B_t$ is the available capacity of the knapsack at time $t$. Note that IIK is known to be NP-hard.
\label{def_IIK}
\end{definition}
\fi

\def\PR{\textrm{StarPR}}
\def\IK{\textrm{IIK}}
\begin{theorem}
The one-layered star case (StarPR) is NP-hard.
\label{general_star_NP}
\end{theorem}
\begin{proof}
What needs to be shown is $\IK \leq_p \PR$ \ifjournal \hspace{-0.35em}. \else \hspace{-0.35em}, where IIK is a known NP-hard problem, the Time-Invariant Incremental Knapsack Problem \cite{incremental_knapsack}.\fi 

Given an instance of IIK, an instance of StarPR is constructed as follows. We construct a graph with $v_i$'s that corresponds to each item $x_i \in X$ and a special node $v$. Edges are added so that each $v_i$ is adjacent to $v$: $E = \{(v, v_i)\}$. Formally, $N = (\{v\} \cup \{v_i\}, E)$. The set of failed nodes $F$ consists of $v_i$'s. The demand $d$ and utility $u$ functions are defined using the given weight $w$ and value $a$ functions, respectively. The available resource function value $r(t)$ for time $t$ is defined by the given capacity function $B_t$. This conversion is obviously executed in polynomial time.

Clearly, IIK reaches the optimum if and only if StarPR reaches the optimum, since the objective functions of these two problems are identical with the settings above. The progressive property of StarPR, which accumulates utility over time, is inherited in the property of IIK solutions that $S_i \subseteq S_{i+1}\ (i = 1, ..., T-1)$.
\end{proof}

Therefore, the PR problem is, in general, a NP-hard problem. This proof also implies that the intractability of a progressive recovery problem changes, depending on the $d$, $u$, and $r$ functions. The work \cite{interdepprogrecov_ealier} provides a polynomial time optimum algorithm for the one-layered star case (Case 1 in \cite{interdepprogrecov_ealier}) with $r: t_i \mapsto C$ and $d: V \rightarrow C$, where $C$ is a constant. 

\subsection{Relations among PR with Different Topology \label{problem_conversion}}

\begin{table*}[t]
\centering
\caption{Summary of Theoretical Results: Theorems collaboratively claims that (1) the special case (the PR in one-layered graphs) captures the PR in general, and (2) a set of meaningful actions is characterized by resource concentration and adjacency between nonfunctional and functional nodes.}
{\renewcommand\arraystretch{1}
\begin{tabular}{p{0.3\textwidth}|p{11em}|p{0.16\textwidth}|p{0.22\textwidth}}
\hline
\textbf{Question}	&\textbf{Lemma/Theorem}		&\textbf{Key Assumptions}	&\textbf{Reasoning}\\
\hline
Should we assign the resources at a step to one node or multiple nodes by splitting them?&Lemma \ref{concentating_is_good_in_stars}: Concentration to one node is better.&One-layered star graphs. &Earlier recovery starts incrementing the utility earlier.\\
\hline
Is it always better to assign the resources at a step to the nodes adjacent to the currently functional nodes?&Lemma \ref{saturation_is_bad}: Assignments to adjacent nodes are better.&One-layered tree graphs.&Saturation does not contribute to the utility.\\
\hline
\multirow{2}{*}{\parbox[t]{0.3\textwidth}{Does the previous two statements hold for more general cases?}}&Theorem \ref{optinrootedtree}&One-layered tree graphs.&Lemma \ref{concentating_is_good_in_stars}, \ref{saturation_is_bad}\\

	&Theorem \ref{optintree}&One-layered graphs.&Theorem \ref{optinrootedtree}, Definition \ref{pseudo_star}\\
\hline
Can we convert the PR in general into the PR in a one-layered graph?&Theorem \ref{general_are_onelayer}&$u(v) = 0\ \forall v \in V_0.$&Lemma \ref{support_pair_ordering}-\ref{no_two_consecutive_zeros}: The order of recovery is not affected even when adding an aggregated node $x$ that makes a graph one-layered. (See Figure \ref{aggregate_x}.)\\
\hline
\multirow{2}{*}{\parbox[t]{0.3\textwidth}{How difficult is the PR problem?}} & Theorem \ref{general_star_NP} & One-layered star graphs. &Reduction from IIK (Definition \ref{def_IIK}).\\
&The PR is NP-hard. & Any graphs satisfying $u(v) = 0\ \forall v \in V_0$. & Theorem \ref{general_are_onelayer}\\
\hline
\end{tabular}
}
\label{sum_theo}
\end{table*}

This section first characterizes the optimum recovery plan in special types of graphs (one-layered graphs). Also, it is proven that the optimum recovery plan of a general network topology shares the same property with that of one-layered graphs, by showing the conversion of the general case into one-layered graph cases. Table \ref{sum_theo} summarizes the theoretical results discussed in this section.

\ifjournal
\begin{lemma}
The optimum recovery plan $P^*$ for any one-layered star graph only consists of concentrating assignments when $r: t_i \mapsto C\ (\forall t_i)$. (See Figure \ref{concent_fig}.)
\label{concentating_is_good_in_stars}
\end{lemma}
\begin{proof}
\ifjournal
First, we argue that the statement is true for a star graph with two nodes with the assumption that $C = 2$, and the demands of the nodes are divisible by $C$. Suppose $P$ only consists of concentrating assignments and $P'$ includes some splitting assignment. 

Because $P$ concentrates resources on a node $v_i$, the node becomes functional after $\frac{d(v_i)}{C (=2)}$ steps. After these steps, it takes $\frac{d(v_j)}{2}$ additional steps to recover the other node $v_j$. Note that during these $\frac{d(v_j)}{2}$ steps, the network-wide utility is always $u(v_i)$. Therefore, $U_P = \frac{d(v_j)}{2} \times u(v_i) + u(v_j)$.

Consider $P'$, which contains a splitting assignment at one time step $t_k$ and concentrating assignments for the other steps. The splitting must be conducted before $v_i$ becomes functional, since there are only two nodes. Then, it takes $\frac{d(v_i)}{2} + 1$ steps for $v_i$ to be recovered and $\frac{d(v_j)}{2} - 1$ steps for $v_j$ to be recovered. Note that $v_j$ receives one unit of resource at both step $t_k$ and step $(\frac{d(v_i)}{2} + 1)$. Therefore, $U_{P'} = (\frac{d(v_j)}{2} -1) \times u(v_i) + u(v_j) < U_P$. The same discussion can be applied to the cases with more splitting. Thus, $U_{P'}$ decreases when more splitting assignments are included in $P'$. When $P'$ only consists of splitting assignments, it takes $d(v_i)$ steps for both nodes to be recovered. Therefore the network-wide utility is $u(v_i) + u(v_j) < U_P$.


Second, we relax the settings by allowing more general demands $d(v_i), d(v_j) \in \mathbb{N}$. 
Without loss of generality, suppose $d(v_i) > d(v_j)$. There are three recovery plans to be compared. Let $P_l$ denote the recovery plan only consisting of concentrations with the prioritization of $v_l$ and $P'$ be a plan including splitting. Based on the previous discussion, $U_{P_i} =  \ceiling{\frac{d(v_j)}{2}} \times u(v_i) + u(v_j)$, and $U_{P_j} =  \ceiling{\frac{d(v_i)}{2}} \times u(v_j) + u(v_i)$.

When $P'$ uses the splitting assignment at one step, $v_j$ is recovered at step $\ceiling{\frac{d(v_j)}{2}}$, and it takes $\ceiling{\frac{d(v_i) - 1 - e}{2}}$ additional steps to recover $v_i$, where $ e = d(v_j)\mod 2$. This is because the splitting assigns one unit of resources to $v_i$, and the ceiling function at step $\ceiling{\frac{d(v_j)}{2}}$ may assign another excess unit, depending on if $d(v_j)$ is divided by $C$. Therefore, $U_P$ is at most $\ceiling{\frac{d(v_i) - 1}{2}}\times u(v_j) + u(v_i) \leq U_{P_j}$. When $P'$ exploits more splitting assignments, the network-wide utility decreases as observed in the previous setting.

It is easily shown by similar discussion that, for any $C (> 2)$, a recovery plan that only includes concentrating assignments is better than plans including splitting assignments. This is because the difference in resource amounts is just a problem of scaling of $C$ and $d$. Thus, the inherent property of the spitting and concentrating assignments does hold even with any different $C$.

It is also obvious that similar discussions hold for general star graphs with $n$ nodes. The key property here is that the splitting delays recovery of a certain node by assigning resources to more nodes, even though the number of steps required to recover all nodes is fixed: $\ceiling{\frac{\sum_{v \in V} d(v)}{C}}$.
\end{proof}

\begin{lemma}
The optimum recovery plan $P^*$ for any one-layered rooted tree never saturates any node that is \underline{not} adjacent to a functional node; i.e., the candidate nodes for resource assignments are always adjacent to a functional node when $r: t_i \mapsto C\ (\forall t_i)$.
\label{saturation_is_bad}
\end{lemma}
\begin{proof}
\ifjournal
For contradiction, consider the case where saturation gives us better network-wide utility. Suppose there are two adjacent nodes $v_i, v_j$ in a rooted tree, such that $v_i$ is adjacent to an independent node, but $v_j$ is not. 

First, we consider the case only with concentrating assignments. After saturating $v_j$, it takes $\ceiling{\frac{d(v_i)}{C}}$ steps to recover $v_i$. During these steps, the utility provided by $v_j$ remains 0. In contrast, when $v_i$ is recovered before $v_j$, it takes $\ceiling{\frac{d(v_j)}{C}}$ to recover $v_j$, and $v_i$ will provide utility of $u(v_i)$ at each of these steps. This generates contradiction, since the number of total steps in both scenarios stays the same.

Second, let us try to improve the total utility, by introducing the splitting assignments, from $\ceiling{\frac{d(v_j)}{C}} \times u(v_i) + u(v_j)$. However, this is impossible based on the discussion in star graphs. When exploiting the splitting at one step, the duration that $v_i$ is functional is strictly less than $\ceiling{\frac{d(v_j)}{C}}$.
\end{proof}

\fi

\begin{theorem}
The optimum recovery plan $P^*$ for any one-layered rooted tree only consists of concentrating assignments that allocate resources to nodes adjacent to a functional node when $r: t_i \mapsto C\ (\forall t_i)$. (See Figure \ref{adj_fig}.)
\label{optinrootedtree}
\end{theorem}
\begin{proof}
\ifjournal
When a network has only one functional node, Lemma \ref{saturation_is_bad} eliminates the possibilities to assign resources beyond the neighbors of the functional node. Then, the network can be considered as a star graph consisting of the functional node and its neighbors. Hence, the statement holds because of Lemma \ref{concentating_is_good_in_stars}.

Accordingly, the node that becomes functional next is adjacent to a functional node. By contracting the edge between the two functional nodes, the problem is reduced to the original problem with one functional node.
\end{proof}

\ifjournal
\begin{definition} \textbf{Pseudo star graph $S_G(\alpha_k)$}: Given a graph $G = (V, E)$ and a node state function $\alpha_k$ at time $t_k$, the logical star graph $S_G(\alpha_k) = (V(S_G(\alpha_k)), E(S_G(\alpha_k))$ consists of one logical functional node $s$ and the nodes adjacent to any of the functional nodes in original graph, and edges connecting $s$ and the others.
\ifjournal
 Formally, $V(S_G(\alpha_k)) = \{s\} \cup \{v_i \in V | \exists v_j \in V \ \mathrm{ s.t. }\ \alpha_k(v_j) = 1 \ \mathrm{ and }\ (v_j, v_i) \in E\}$, and $E(S_G(\alpha_k))) = \{(s, v_i) \mid v_i \in V(S_G(\alpha_k))\}$. \label{pseudo_star}
\fi
\end{definition}
\fi

The same statement holds for the case where $G_0$ has more nodes, and there exists more biconnected pairs of nodes between $G_0$ and $G_1$.
\begin{theorem}
For any one-layered graph, the optimum recovery plan $P^*$ only consists of concentrating assignments that allocate resources to nodes adjacent to a functional node when $r: t_i \mapsto C (\forall t_i)$.
\label{optintree}
\end{theorem}
\begin{proof}
\ifjournal
It is trivial that the optimum recovery plan does not saturate any node that is not adjacent to a functional node, even when a graph has more than one independent nodes or any cycle. Based on a discussion similar to Lemma \ref{saturation_is_bad}, an assignment of resources to a node adjacent to a functional node always provides more network-wide utility over time, since the node assigned resources starts contributing to the utility in an earlier step. Thus, the candidate nodes for resource assignment at each step $t_k$ are the nodes adjacent to any functional node.

Therefore, a resource assignment decision at each time step, $P[t_k]$ is equivalent to the progressive recovery problem in a logical star graph $S_G(\alpha_k)$, where $\alpha_k$ is a node state function reflecting recovery from $t_0$ to $t_{k-1}$. Therefore, it can be considered as the recovery problem in a star graph with a single logical functional node at the center and surrounding leaf nodes $V(S_G(\alpha_k))$. 

Hence, it is easily provable, by the argument in Lemma \ref{concentating_is_good_in_stars}, that the optimum plan does not involve splitting assignments, since the concentration of the split resources to a node can always recover the node in an earlier time step and provide more network-wide utility.
\else
The multiple nodes in $G_0$ can be recognized as one logical node, since there is no need to assign resources to them. Hence, the resource assignment decision at each time step is equivalent to the progressive recovery problem in a logical star graph. Please see our technical report \cite{ourtechrep} for details.
\fi
\end{proof}

\def\conditionA{$d(v_i) + d(v_j) > 2C - 1\ (\forall (v_i, v_{j \neq i}) \in V \times V)$}

\ifjournal

Next, we claim that the progressive recovery problem with any network topology can be converted into the case in a one-layered graph.

\begin{figure}[t]
\def\scale{0.83}
\centering
\begin{minipage}{0.48\columnwidth}
\centering
\includegraphics[width=\scale\columnwidth]{\figpath 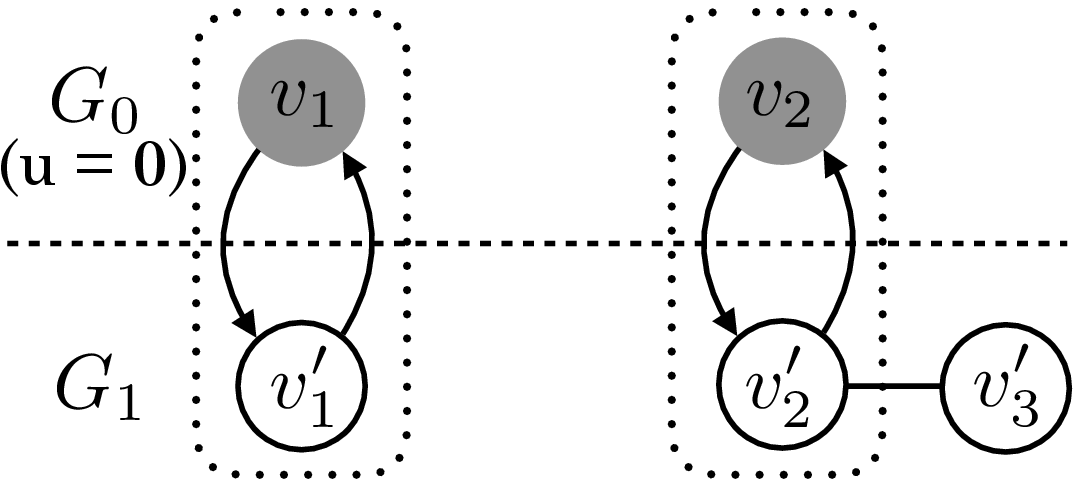}
\caption{Supporting Pairs $(v_1, v'_1)$ and $(v_2, v'_2)$: The first recovery occurs only when two nodes in a pair are saturated.\label{supportpair}}
\end{minipage}\hfill
\begin{minipage}{0.48\columnwidth}
\centering
\includegraphics[width=\scale\columnwidth]{\figpath 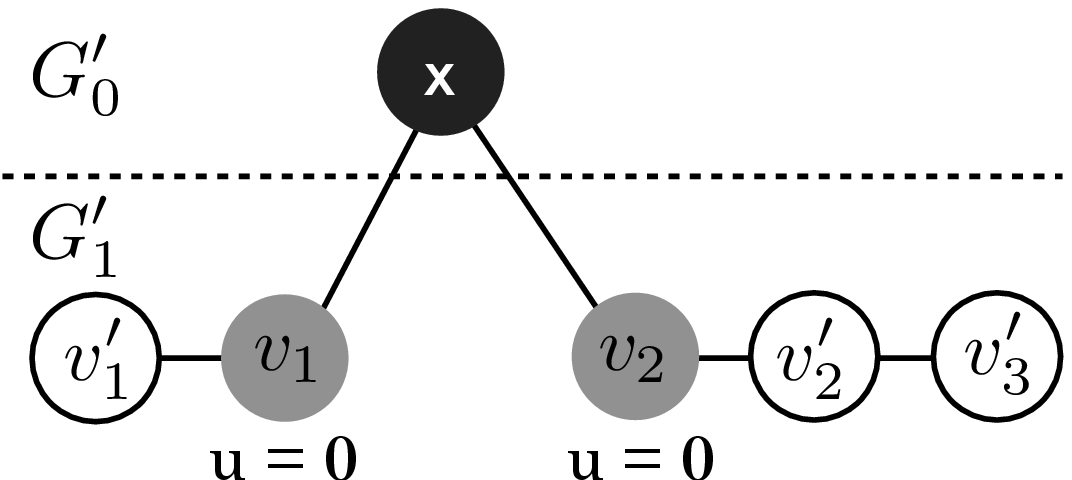}
\caption{Conversion into a One-layered Graph: The graph in Figure \ref{supportpair} is converted into a one-layered graph. A new node x logically forms a new layer $G'_0$, and the rest of the nodes form the other new layer $G'_1$.\label{aggregate_x}}
\end{minipage}
\end{figure}
\def\scale{1}

\begin{definition}
A pair of nodes $v \in V_0$ and $v' \in V_1$ is called a \textit{support pair} when $(v, v') \in A_{01}$ and $(v', v) \in A_{10}$. (See Figure \ref{supportpair}.)
\end{definition}

\begin{lemma}
When $v$ and $v'$ are the first support pair recovered in a given graph $N$, the order of saturation of these two nodes does not influence the total utility.
\label{support_pair_ordering}
\end{lemma}
\begin{proof}
\ifjournal
Let us assume that a recovery plan saturates $v$ first and $v'$ later. Note that there may be some nodes saturated before and between $v$ and $v'$. Since $v, v'$ are the first supporting pair to be recovered, there is no functional node in $N$ before $v'$ is saturated. The total utility generated until the step $t_i$ when $v'$ is saturated is $u(v) + u(v') + \sum_{w \in V_r} u(w)$, where $V_r \subseteq \Satu[t_i]$ is a set of saturated nodes that are reachable from $v$ or $v'$. When we exchange the ordering of $v$ and $v'$, the total utility until the step $t_i$ when $v$ is saturated remains the same, because the saturated nodes until $t_i$ are same. Therefore, the order of saturation of $v$ and $v'$ does not change the total utility.
\end{proof}

\begin{lemma}
In any graph, the first two nodes saturated by the optimum recovery plan $P^*$ are always the nodes in a support pair.
\label{support_pair_first}
\end{lemma}
\begin{proof}
\ifjournal
For contradiction, assume a node $w \in V_0$ was a node saturated at first by the optimum recovery plan $P^*$, and the two nodes $v, v'$ in a support pair will be recovered right after $w$. Without loss of generality, it is assumed that $v$ is saturated first from Lemma \ref{support_pair_ordering}. Then, the total utility until the step $t_i$ when $v'$ is saturated is $u(v) + u(v) + \beta u(w)$, where $\beta = 1$ iff $w$ is adjacent to $v$ or $v'$; otherwise, 0. 

However, another recovery plan $P'$, which saturates $v$ and $v'$ first and $w$ later, provides the total utility until $t_i$ of $2(u(v) + u(v')) + \beta u(w)$, since $v$ and $v'$ are already functional at $t_{i-1}$. This contradicts the fact that $P^*$ is the optimum.
\end{proof}

\begin{lemma}
In the one-layered rooted tree where any node adjacent to the node $u \in V_0$ has utility of zero: $d(v) = 0\ (\forall v \textrm{ s.t. } (u, v) \in A_{01})$, the second node $v_2$ recovered by the optimum recovery plan $P^*$ has utility strictly greater than zero: $d(v_2) > 0$.
\label{no_two_consecutive_zeros}
\end{lemma}
\begin{proof}
\ifjournal
All the nodes adjacent to $u \in V_0$ have utility of zero. Therefore, the first node $v_1$ recovered by the $P^*$ is one of these node. For contradiction, assume the second node $v_2$ is also one of these zero-utility nodes, and let $v_k$ be the first node recovered, whose utility is greater than 0 ($k$-th node recovered in the plan).

In order to recover $v_k$, it is necessary to have a zero-utility node that is already recovered for the reachability to $u$. There are two possible scenarios: (1) $v_1$ is adjacent to $v_k$, or (2) $v_j\ (2 \leq j < k)$ is adjacent to $v_k$.

For the first scenario, we can exchange the recovery order of $v_2$ and $v_k$. This exchange has no influence on the candidate nodes at each step after $k$-th recovery, because the recovered nodes until $k$-th recovery stay the same. However, it increases the utility and contradicts the fact that $P^*$ is optimum.

For the second scenario, we can exchange the recovery order of $v_1$ and $v_j$. Again, this does not change any candidate sets for recovery after $k$-th recovery. Since $v_j$ is recovered at the very beginning, we can use the same discussion with the first scenario. Therefore, it provides a contradiction. Therefore, the second node recovered in the optimum plan should have nonzero-utility.
\end{proof}

\begin{figure*}[t]
\centering
\begin{minipage}{0.24\textwidth}
\includegraphics[width=\textwidth]{\figpath 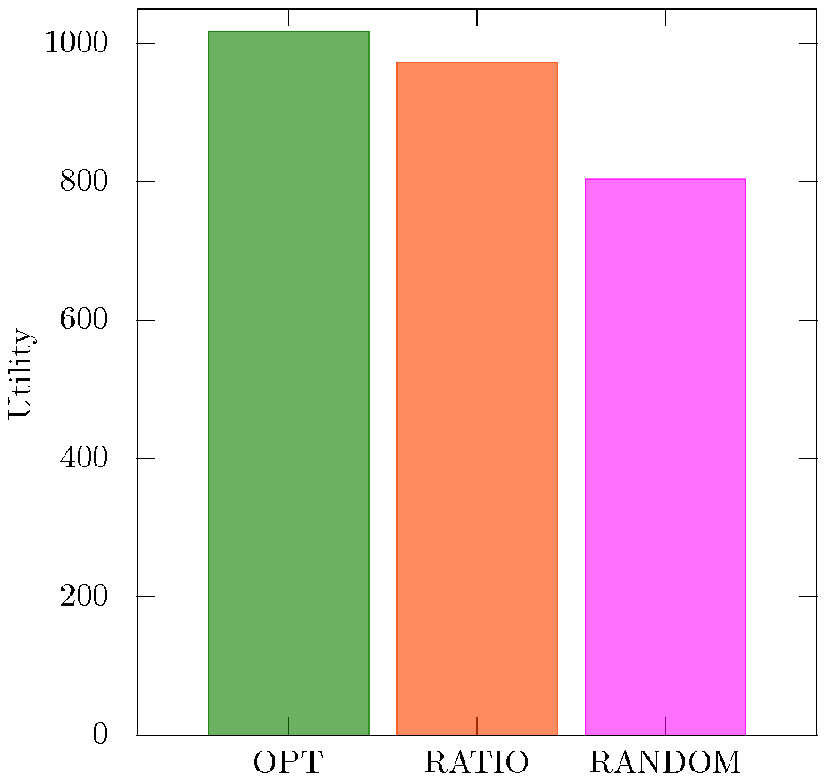}
\caption{Total Utility under Compliant Settings (the IBM Network): RATIO could achieve the near-optimal solutions.\label{ibm_general_result}}
\end{minipage}\hfill
\begin{minipage}{0.24\textwidth}
\includegraphics[width=\textwidth]{\figpath 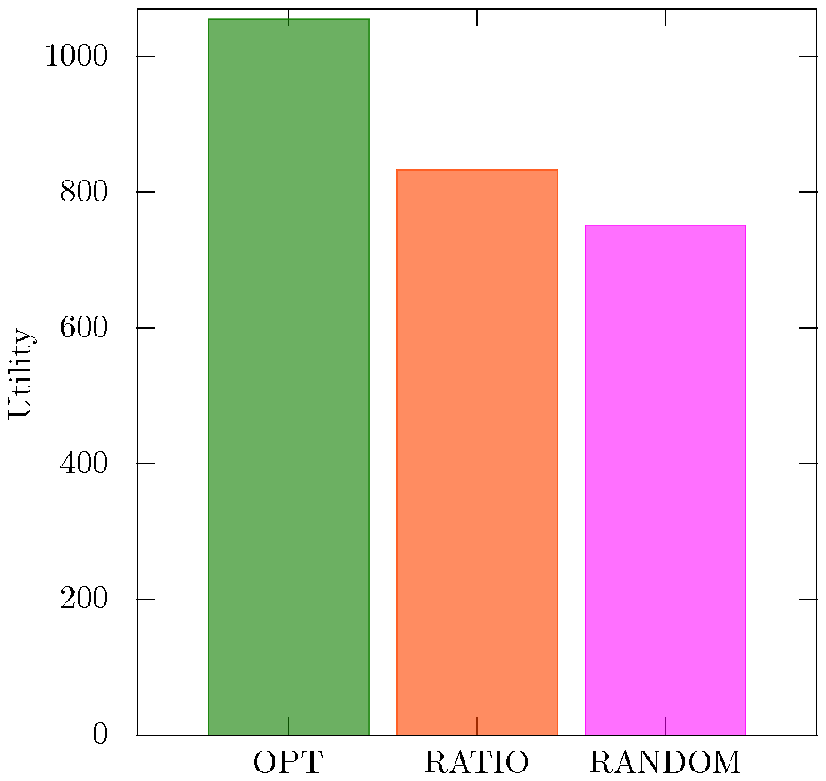}
\caption{Total Utility under Adversarial Settings (the IBM Network): The performance of RATIO is easily deteriorated with small changes of the attributes of some nodes. \label{ibm_ad_result}}
\end{minipage}\hfill
\begin{minipage}{0.48\textwidth}
\centering
\includegraphics[width=0.9\textwidth]{\figpath 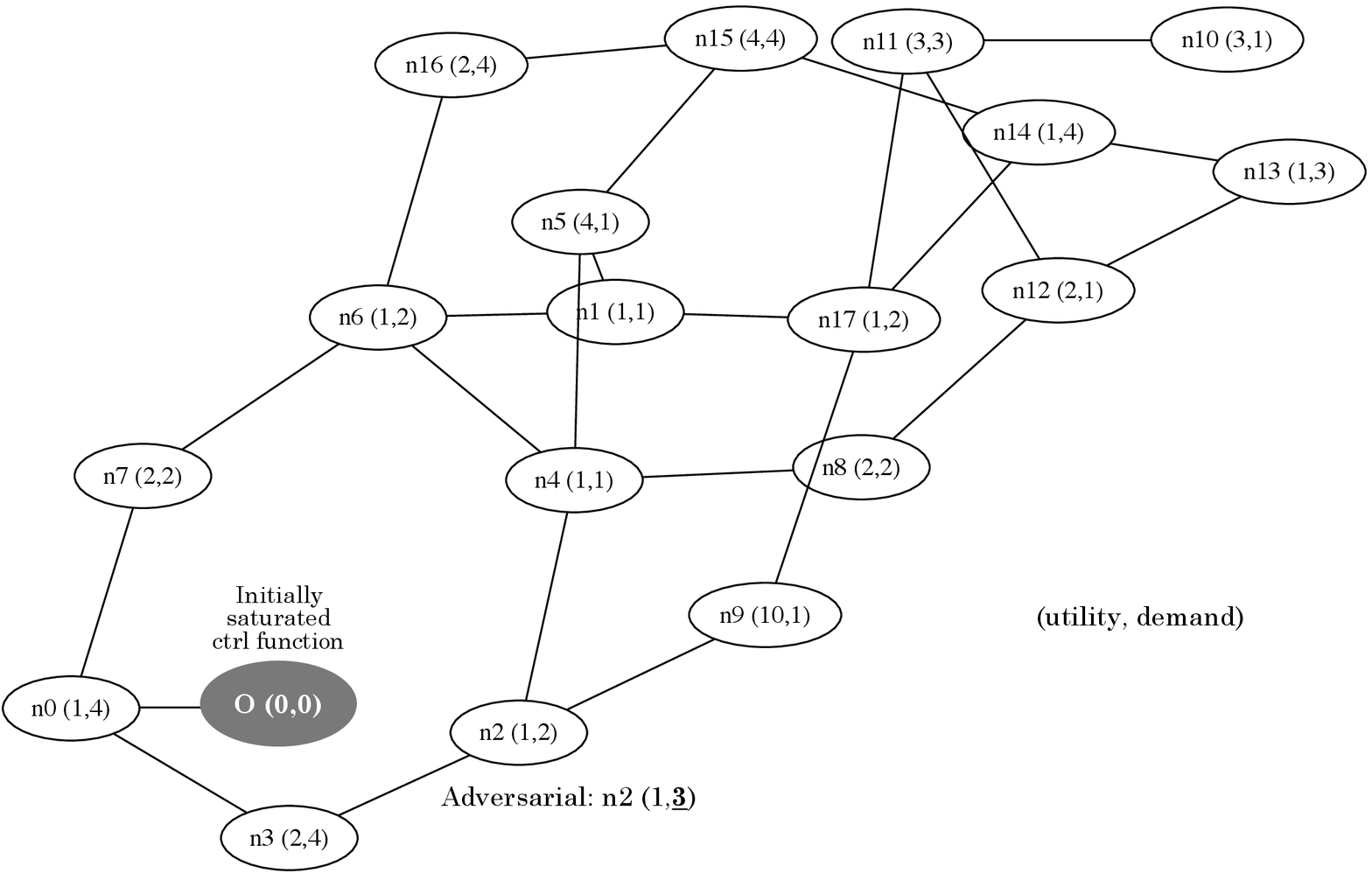}
\caption{The IBM Network with Node Attributes for Preliminary Simulations: The node attributes inside of each node are used for the compliant scenario. In an adversarial scenario, node n2 is intentionally attacked to make its demand 3. \label{ibm_settings}}
\end{minipage}
\end{figure*}

\begin{theorem}
A progressive recovery problem with any general graph with $u(v \in V_0) = 0$ has an equivalent progressive recovery problem with a one-layered graph. 
\label{general_are_onelayer}
\end{theorem}
\begin{proof}
\ifjournal
The problem with a general graph is converted into the problem with a one-layered graph as follows. We add a new node $x$ to $G'_0$ and put all the nodes and edges in the original $N$ into $G'_1$. An edge is added between $x \in V'_0$ and each $v \in \hat{V}'_1$, where $\hat{V}'_1$ consists of nodes that are originally in $V_0$ of $N$; i.e. $u(v \in \hat{V}'_1) = 0$. Figure \ref{aggregate_x} illustrates an example of constructing a one-layered graph from the graph shown in Figure \ref{supportpair}.

Lemma \ref{support_pair_first} shows that the first two nodes to be saturated (recovered) are the ones in a support pair. Also, according to Lemma \ref{support_pair_ordering}, it can be assumed without loss of generality that a node $v$ in $V_0$ in each support pair is the first node to be saturated. 

The edges newly added confirm that the first node recovered is one of the nodes in $V_0$, since $x$ is the only saturated node in the initial step. The other correspondence between two problems to be checked is that the second node recovered in $N'$ is $v'$ that forms a support pair with $v$ in the original graph $N$, and Lemma \ref{no_two_consecutive_zeros} guarantees this. 
\else
It follows from Lemma \ref{support_pair_ordering}-\ref{no_two_consecutive_zeros}. Please see our technical report \cite{ourtechrep} for details.
\fi
\end{proof}

Therefore, it is enough to think about the cases of one-layered graphs. Also, it is possible to aggregate multiple nodes in $G_0$ into one logical node in $G_0$ to decide the resource assignment, as the proof of Theorem \ref{optintree} suggests. Thus, without loss of generality, the rest of this paper only deals with the one-layered graphs with one node in $G_0$.

\section{Heuristic for Progressive Recovery and Its Limitation \label{ratio_section}}
This section describes how a simple heuristic algorithm named RATIO performs under (1) compliant failure scenarios into which most of random failures fall and (2) intentional failure settings that could be made by a small change to the compliant cases. Note that the term \textit{compliant} is used to describe the situations where there is no adversarial weight settings defined in Section \ref{limitofratio}.

\subsection{RATIO Heuristic and its Performance}
RATIO is a greedy heuristic algorithm inspired by the approximation algorithm of the set cover problem. This heuristic assigns resources to the most cost-effective nodes among the nodes adjacent to functional nodes at each time step by calculating $\frac{u(v)}{d(v)}$. Algorithm \ref{ratio_pseudo} shows the pseudo code of RATIO.

In order to understand the performance of RATIO, a preliminary simulation has been conducted using the IBM network \cite{ibmgraphcite}. Figure \ref{ibm_settings} shows the settings of node attributes used for the simulation. In the simulation, it is assumed that one unit of resource becomes available at each time step, and n0 is the only node that is connected to the initially saturated control function node $O$ that belongs to $G_0$.

\def\rnd{\mathrm{randomInt}}
\begin{algorithm}[t]
\caption{RATIO$\left(G, F_{t-1}\right)$}
\label{ratio_pseudo}
\begin{algorithmic}[1]
\REQUIRE A graph $G = (V, E)$, A set $F_{t-1}$ of nonfunctional nodes at time $t-1$  
\STATE $W = \mathrm{neighbor}(V \setminus F_{t-1})$
\STATE Sort $W$ based on $\dfrac{u(v_i)}{d(v_i)}$ in the decreasing order
\STATE Perform concentrating allocation from the head of $W$
\end{algorithmic}
\end{algorithm}

Figure \ref{ibm_general_result} illustrates the total utility obtained by each method. Note that OPT is the optimum total utility, which is calculated by DP-OPT in Appendix \ref{dpopt}, and RANDOM is a heuristic algorithm that randomly selects one of the nonfunctional nodes adjacent to functional nodes. As will be understood, RATIO achieves a near-optimal result with the compliant setting. While we do not include the results in other topologies here to avoid redundancy, RATIO shows similar performances in other preliminary simulations conducted with compliant settings. This fact will be reviewed in the evaluation section (Section \ref{eval_section}) again.


\def\scale{0.27}
\begin{figure}[t]
\centering
\includegraphics[width=\scale\textwidth]{\figpath 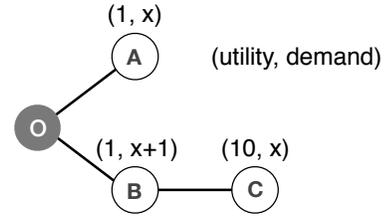}
\caption{An Adversarial Toy Example: The worse utility-demand ration of node B hides, from RATIO, node C that potentially produces higher overall utility, even when compensating for the loss by selecting node B.\label{adversarialeg}}
\end{figure}

\subsection{Adversarial Example and the Limitation \label{limitofratio}}
Figure \ref{adversarialeg} illustrates a minimal adversarial setting for the RATIO heuristic. Suppose that one unit of resource is available at each time step $(r=1)$ and node O is a saturated function node. Note that all the other nodes are nonfunctional at the beginning, and their demands are depicted in the figure.

At the first round of recovery, RATIO chooses node A between node A and B, since $\frac{1}{x} > \frac{1}{x+1}$ $(x \geq 1)$. It takes $\frac{x}{r (= 1)}$ time steps to recover $A$. Then, RATIO recovers node B and C in order, which takes $x+1$ and $x$ steps, respectively. Therefore, the total utility of RATIO is always $u(A) \cdot (2x+1+1) + u(B) \cdot (x+1) + u(C) \cdot 1 = 3x +13$.

In contrast, the optimum strategy is recovering node B and C first, and then node A. In this case, the total utility is $u(B) \cdot (2x + 1) + u(C) \cdot (x + 1) + u(A) \cdot 1 = 12x + 12$. When $x$ becomes larger, the total utilities of RATIO and the optimum will diverge more drastically. 

In general, when a node $v$ (node C in Figure \ref{adversarialeg}) has larger utility in a network and all the neighbors of $v$ (node B in Figure \ref{adversarialeg}) show lower effectiveness relative to the other nodes (node A in Figure \ref{adversarialeg}), that part of the network could cause the adversity similar to the previous example.

Since the minimal adversarial example is quite simple, a failure incident in a larger network could contain it as an embedded substructure.
Furthermore, it could be said that RATIO is vulnerable to failure events by malicious attacks, since an attacker can easily embed this substructure in an attack and deteriorate the interim utility arbitrarily by setting $x$ as large as possible. This scenario can be interpreted as an intentional attack where an attacker tries to hide a node with a high utility value $v$ from RATIO by imposing more damages to the neighbors of $v$.

Figure \ref{ibm_ad_result} shows a preliminary result of progressive recovery in the IBM network with an adversarial substructure. The adversarial failure is realized by a change in the demand of n2, as shown in Figure \ref{ibm_settings}. The change could be interpreted as more damage caused by an intentional attack to a specific node. The result indicates that the performance of RATIO is degraded by the adversity. RATIO only achieves 79.0\% of the optimum with the adversarial setting, although it approximately reaches 95.6\% of the optimum with the previous compliant scenario. It is noteworthy that a slight change of the demand of one node (from 2 to 3) can worsen the performance of RATIO to this extent. Based on the argument with the toy example and the preliminary results, it is deducible that RATIO may be worse than RANDOM in some adversarial scenarios. This discussion motivates us to introduce the following deep reinforcement learning-based algorithm that demonstrates the robustness against such adversarial scenarios.

\section{DeepPR: Reinforcement Learning for Progressive Recovery \label{DeepPR_section}}
A Deep reinforcement learning algorithm for the Progressive Recovery problem (DeepPR) is explained in this section. DeepPR has its roots in an established reinforcement learning method called Deep Q-Learning. Here, the key concepts of the technique are summarized. Additionally, we describe how to connect the RATIO heuristic to deep reinforcement learning to improve upon normal exploration methods.

\subsection{Q-Learning}
Reinforcement Learning (RL) is a method to learn the best mapping of states $\mathcal{S}$ to actions $\mathcal{A}$. The key elements of RL include the \textit{agent}, who learns the mapping of state action pairs to numerical rewards for its trial actions, and the \textit{environment}, which updates states and returns the numerical \textit{reward} depending on actions the agent takes.

In Q-learning, the mapping is learned using the \textit{action-value function} $Q: \mathcal{S} \times \mathcal{A} \rightarrow \mathbb{R}$ that represents the quality of each state-action pair. In theory, the Q-value of a state-action pair converges to $Q^*(s,a)$ after infinite trial actions (experiences): $Q^*(s, a) = \max_\pi \mathbb{E}[\sum_{k = 0}^{\infty} \gamma^k r_{t+k} \mid s_t = s, a_t = a, \pi]$, which is the expected reward achievable by following the optimum action sequence (policy) $\pi$ from state $s$ taking action $a$ at time $t$. Note that $\gamma$ is a discount factor for future rewards that defines the learning horizon. 


\subsection{Deep Q-Network (DQN)}
Mnih et al. \cite{mnih2015human} report a significant improvement in RL by introducing the Deep Q-Network (DQN). Instead of explicitly calculating Q-values, a DQN uses neural networks (NNs)---parametrized by a weight function $\theta$---as a function approximator to estimate the optimum Q-values: $Q^*(s, a) \approx Q(s, a; \theta)$.

The dramatic improvement by DQNs in learning performance is achieved mainly by introducing \textit{experience replay} and \textit{Target-Net} \cite{mnih2015human}. The $\epsilon$-greedy exploration method was also employed to effectively explore the state-action space.


\subsubsection{Experience Replay} It is known that the correlation among experiences, which each are represented by quadruples $e_t = (s_t, a_t, r_t, s_{t+1})$ of a state, action, reward at time $t$, as well as a resulting state at time $t+1$, causes fluctuations in the learning process. Experience Replay buffers all the experiences $B = \{e_t\}$ and takes random samples from $B$ for the Q-value updates. This random sampling prevents DQNs from undergoing fluctuation in training due to learning from correlated sequential experiences.

\subsubsection{Target-Net and Eval-Net} In order to stabilize the learning, it is proposed to use two separate DQNs; one named Eval-Net for learning from each sampled experience, and the other, named Target-Net, for calculating the target Q-values. The weight function $\theta_T$ of Target-Net is periodically updated by copying the weight function $\theta$ of Eval-Net.

For each sampled experience $e_t$, the parameters of Eval-Net are updated by any gradient method with respect to the loss function $L(\theta)$, which represents the difference between the Q-values estimated by Eval-Net and Target-Net.
\begin{align*}
L(\theta)= &\mathbb{E}_{e_t \sim U(B)} \biggl[\Bigl( r_t + \gamma \max_{a_{t+1}} Q(s_{t+1}, a_{t+1}; \theta_T)\\
 &- Q(s_t, a_t; \theta)\Bigr)^2\biggr],
\end{align*}
where $e_t \sim U(B)$ indicates the random sampling of $e_t = (s_t, a_t, r_t, s_{t+1})$ from the buffer $B$ by the experience replay.

%

\begin{figure*}[t]
\centering
\begin{minipage}[t]{0.38\textwidth}
\centering
\includegraphics[width=1\columnwidth]{\figpath 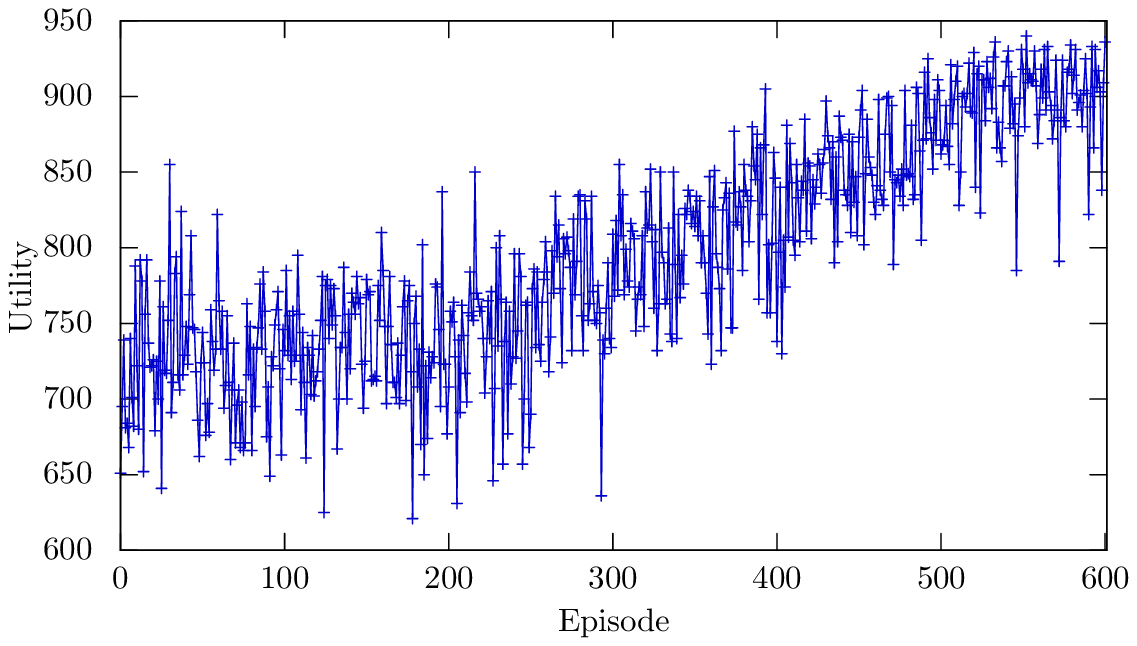}
\caption{Learning Curve of DeepPR: A GNP random graph is given a compliant attribute setting. \label{lc}}
\end{minipage}\hfill
\begin{minipage}[t]{0.38\textwidth}
\centering
\includegraphics[width=1\columnwidth]{\figpath 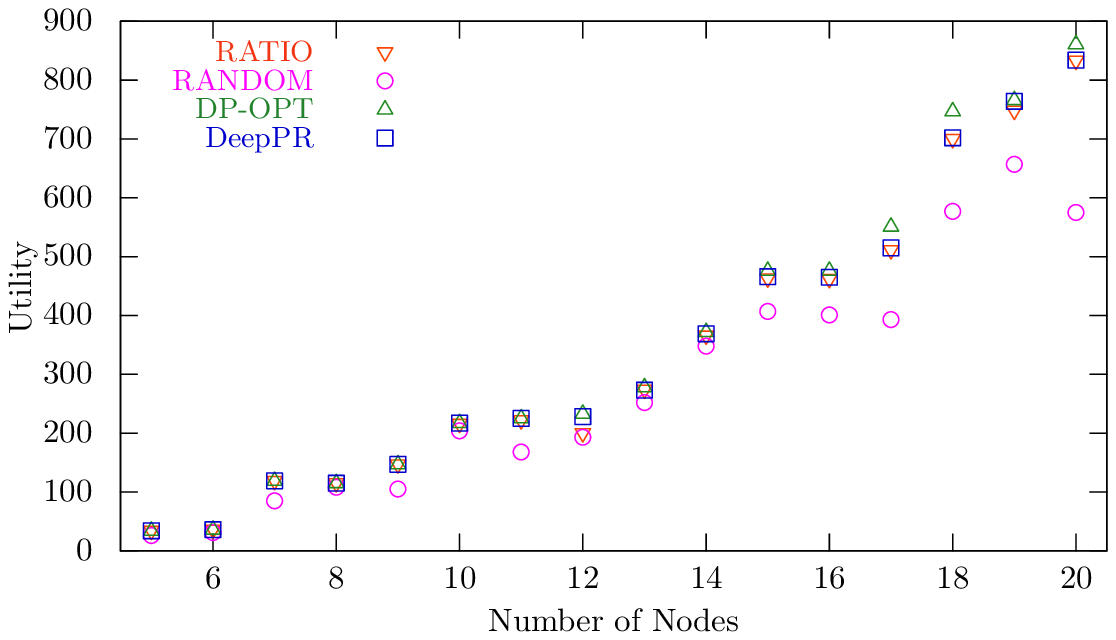}
\caption{Utility under Compliant Settings (GNP Random Graphs): Both DeepPR and RATIO approach the theoretical optimum, while DeepPR demonstrates slightly better results. \label{gnp_result}}
\end{minipage}\hfill
\begin{minipage}[t]{0.21\textwidth}
\centering
\includegraphics[width=1\columnwidth]{\figpath 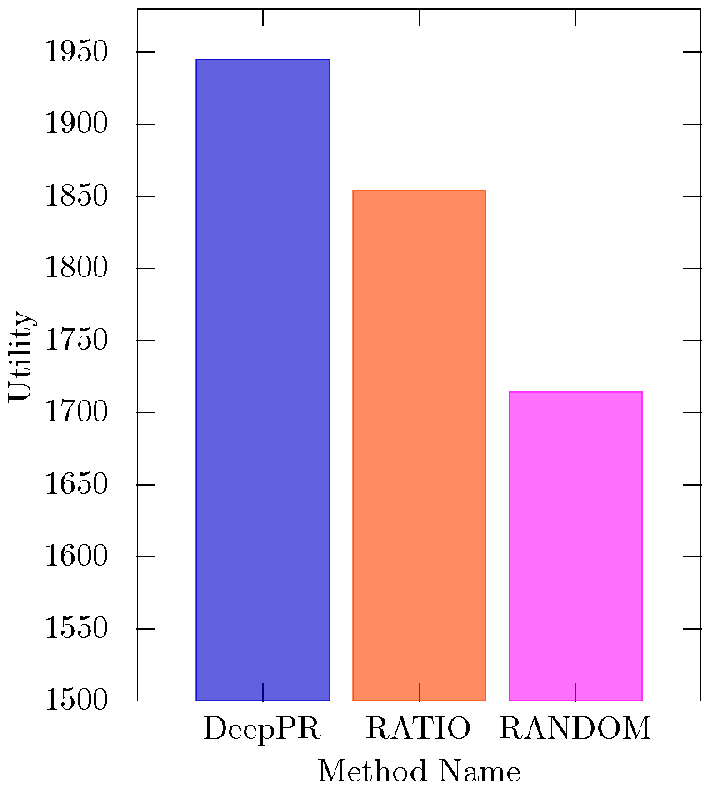}
\caption{Utility under Compliant Settings (the BT North America Graph). \label{BTNA_result}}
\end{minipage}
\end{figure*}

\subsubsection{$\epsilon$-greedy Exploration}  The tradeoff between exploration and exploitation is one of the crucial challenges in RL. The $\epsilon$-greedy exploration is a commonly used approach to address this challenge. In this greedy approach, the agent follows the current best action known in a current state to reinforce the previous learning (exploitation) with probability $(1 - \epsilon)$. With probability $\epsilon$, it tries an exploration by taking an action that is not determined by the previous learning. 

DeepPR integrates two simple algorithms to realize the exploration; namely, RATIO and RANDOM. DeepPR chooses the best action following the RATIO heuristic with a predefined probability $\omega_{\mathrm{RATIO}}$ and selects a random action from a legal action set with probability $(1 - \omega_{\mathrm{RATIO}})$. Therefore, the probability for DeepPR to take an action based on RATIO is $\epsilon\, \omega_{\mathrm{RATIO}}$.

\subsection{Applying DQN to PR}
In our problem, the agent tries to learn the optimum resource allocations to nonfunctional nodes. Therefore, the legal actions for our agent are selecting a subset of nonfunctional nodes. Here, we assume a situation where at most one node is fully recovered at a time step by setting \conditionA, where $C$ is the amount of available resources at a time step. Therefore, each \textbf{action} at time step $t_k$ is represented as an ordered pair of nodes $[v_i, v_j]$, where the first node $v_i$ is assigned $\min\{r(t_k), d(v_i)\}$ resources, and the second node $v_j$ receives the remaining resources if they exist. The number of legal actions is always the number of 2-permutations of $V$, $P(|V|, 2)$. Each \textbf{state} is represented as a $(|V| \times 1)$ vector in which $i$th element indicates the remaining demand of the corresponding node $v_i \in V$. The \textbf{reward} of a state-action pair is the sum of utilities of the functional nodes.

One of the biggest challenges in our problem is the size of the state space, which grows exponentially in the number of nodes. For example, a graph consisting of 20 nodes with a minimal demand setting, where $d(v) \in \{0, 1\}\ (\forall v \in V)$, has over one million ($ \approx 2^{20}$) possible states, and the number of state-action pairs is approximately $2^{20} \times P(20, 2)$. In order to improve the performance of exploration, the integrated exploration, which comprises RANDOM and RATIO, is adopted in DeepPR with appropriate $\epsilon$ and $\omega_{\mathrm{RATIO}}$.



\section{Evaluations \label{eval_section}}
\ifjournal
Simulations are conducted with different topologies and node attributes, as explained below. DeepPR is also evaluated in both compliant and adversarial failures, being compared with the theoretical optimum (DP-OPT), RATIO, and RANDOM.
\fi

\subsection{Simulation Settings}
\subsubsection{Network Topology} 
GNP random graphs \cite{networkxGNP}, the BT North America graph \cite{BTNAmerica}, and the IBM graph \cite{ibmgraphcite} are used as network topologies. Since our theoretical results indicate it is enough to test the algorithm performance in one-layered graphs with single node in $G_0$, a node in $G_0$ is randomly selected in each graph. 
For GNP random graphs, the following ranges are used: $p = 0.2$, and $n \in \{5, 6, ..., 20\}$ for compliant scenarios and $n \in \{5, 6, ..., 34\}$ for adversarial settings. Note that only connected GNP random graphs are fed into our simulations. The BT North America graph is based on an IP backbone network with 36 nodes and 76 edges. The IBM graph is a backbone network consisting of 18 nodes and 16 edges. 


\subsubsection{Node Attributes and Available Resource} The utility, demand, and resource values are randomly selected among the integers within given ranges for GNP random and the BT North America graphs. Here, the following setting is used: (utility range, demand range, resource amount available at each time step)$\, = ([1,4], [1,2], 1)$. 
The setting for the IBM graph follows the node attributes shown in Figure \ref{ibm_settings}. 

Also, it is assumed that all the nodes in a given network are initially nonfunctional: $F[t_0] = V$.






\subsubsection{DQN Settings}
Our DQN consists of three fully connected layers: the input, middle, and output layer with $|V|$, 200, and $P(|V|, 2)$ neurons, respectively. The input layer receives the state vector, which represents the remaining demands, and the output layer indicates the evaluation of possible legal actions for a given state. The Rectified Linear Unit (ReLU) is used for the activation function, and the reward discount factor $\gamma$ is set to 0.6. The training of the DQN is conducted by the Adam algorithm (AdamOptimizer in TensorFlow \cite{tf}) that minimizes the estimation loss $L(\theta)$.

For exploration, $\epsilon$ is initially set to 1.0, decreasing by 0.0001 after every episode until 0.1. This encourages DeepPR to visit more diverse state-action pairs at the beginning and to reinforce its learning as it has experienced more. Additionally, $\omega_{\mathrm{RATIO}}$ is fixed to 0.5 to incorporate RATIO in exploration.

\begin{figure*}[t]
\centering
\begin{minipage}[t]{0.16\textwidth}
\end{minipage}\hfill
\begin{minipage}[t]{0.28\textwidth}
\centering
\includegraphics[width=\columnwidth]{\figpath 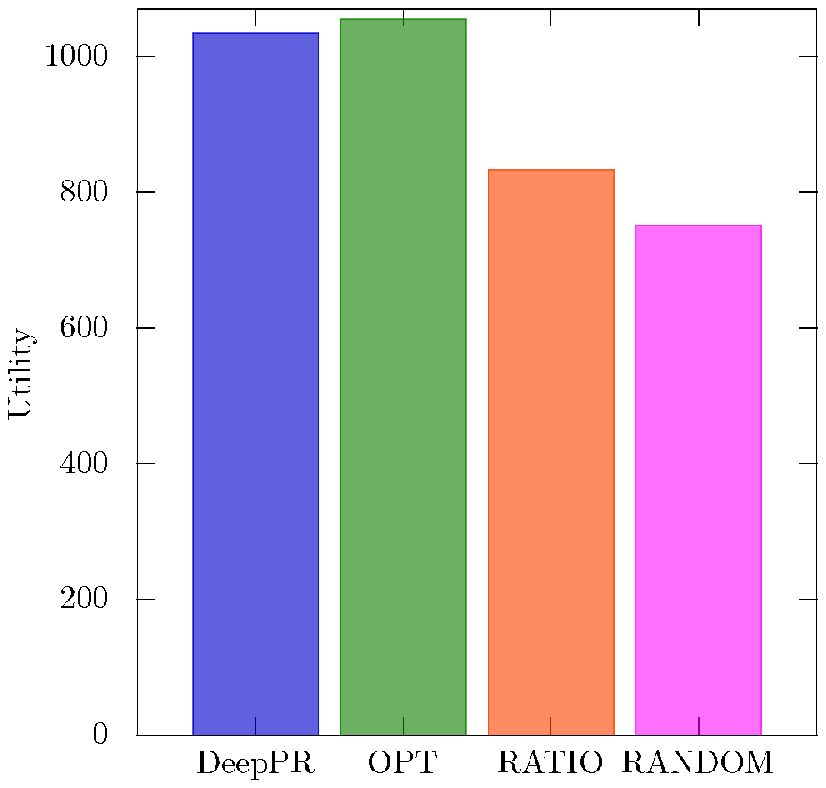}
\caption{Utility under Adversarial Settings (the IBM Graph): DeepPR shows the performance similar to what it demonstrated with compliant scenarios, although RATIO suffers from the adversarial attack.\label{IB_ad_deep}}
\end{minipage}\hfill
\begin{minipage}[t]{0.47\textwidth}
\centering
\includegraphics[width=\columnwidth]{\figpath 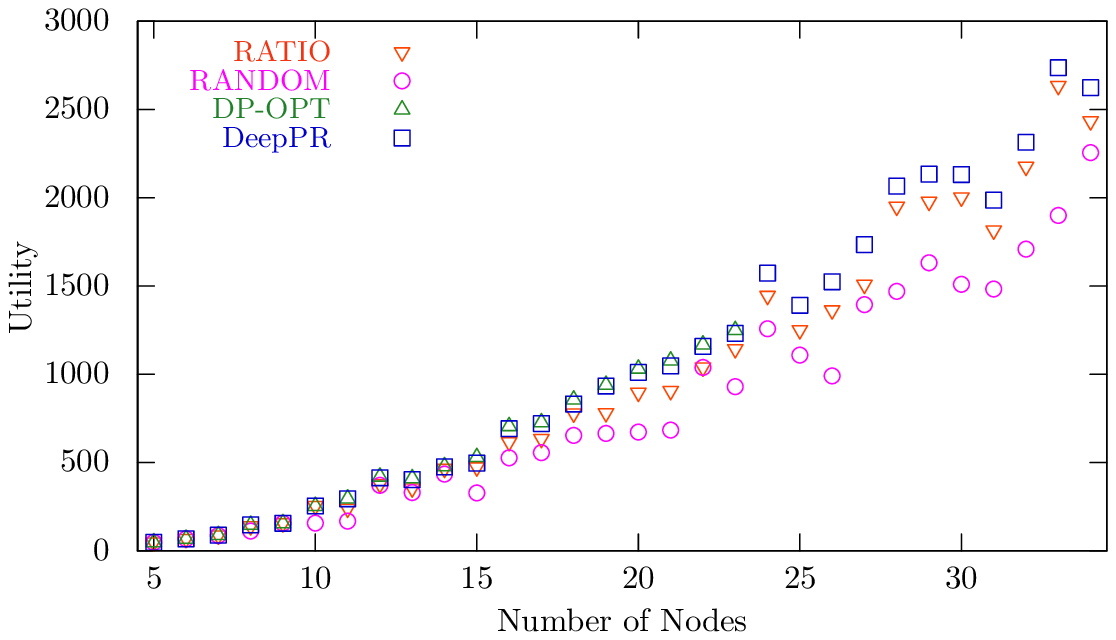}
\caption{Utility under Adversarial Settings (GNP Random Graphs): DeepPR obtains the solutions closer to the optimum, while the difference between DeepPR and RATIO increases from the compliant cases.\label{result_gnp_ad}}
\end{minipage}\hfill
\begin{minipage}[t]{0.16\textwidth}
\end{minipage}
\end{figure*}

 \subsection{Simulation Results}
Figure \ref{lc} illustrates a sample of the learning curve of DeepPR over episodes, which are alternating sequences of states and actions from the initial network state to the fully recovered state. This sample is obtained in a GNP graph with 19 nodes, and similar curves are also observed in other graphs. Since the NNs are randomly initialized, the initial $Q$-values do not reflect the actual rewards. 
Through the update on $Q$-values and explorations, the NNs are trained to select an action that maximizes the total utility.
In the figure, the utility (total reward) that DeepPR achieves stays at approximately 725 until around the 250th episode, and after that, it continues increasing towards around 900. Because of the exploration by random actions, utility values fluctuate during the entire training period. Note that each episode takes 1.057 seconds on average in a computer with a 2.5 GHz Intel Core i5 CPU, Intel HD Graphics 4000 (1536 MB), and 8 GB memory.

Figure \ref{gnp_result} indicates a comparison among the four algorithms in terms of total utility in GNP random graphs. In smaller graphs, the utility obtained by DeepPR always matches with the theoretical optimum (DP-OPT). In theory, $Q$-learning is guaranteed to achieve the optimum by visiting each state-action pair an infinite number of times. Since it is easier to visit each state-action pair a greater number of times in graphs with fewer states and action choices, the estimation of $Q$-values seems to converge to more accurate values, which leads to the optimum. In contrast, the difference between DP-OPT and DeepPR increases in some larger graphs for the same reason. Compared to RATIO, DeepPR performs slightly better in those larger graphs. DeepPR achieves 96.6\% of the optimum on average in the graphs of size 15 to 20 nodes, while RATIO reaches 95.8\% of the optimum in the same graphs. Also, RANDOM is the worst heuristic among the four methods over all sizes of graphs and continues getting worse along with the graph size because of the increase of legal actions. 

Figure \ref{BTNA_result} shows the utility obtained by three algorithms in the BT North America graph. Here, DP-OPT is not included since it is intractable due to the number of nodes. In this practical topology, we also observed a trend similar to the results from GNP graphs. As mentioned in Section \ref{ratio_section}, all simulation results indicate that RATIO could attain the utility values close to the utility obtained by DeepPR and the optimum when input node attributes do not contain an adversarial substructure.

In contrast, the total utility obtained by RATIO is easily deteriorated by intentional attacks with the adversarial settings as shown in Figure \ref{IB_ad_deep} and \ref{result_gnp_ad}. However, DeepPR demonstrates robustness in these scenarios; i.e., DeepPR keeps achieving the solutions closer to the optimum as it does in the previous compliant cases. For example, the total utility obtained by DeepPR in the GNP graph with 23 nodes is 98.5\% of the optimum, though RATIO obtains 91.4\% of the optimum. This is because DeepPR still explores other possible state-action pairs by random exploration with probability $\epsilon(1-\omega_{\mathrm{RATIO}})$. The difference in performance between DeepPR and RATIO diverges in GNP random graphs with the adversarial settings, as shown in the figure.

 
%

\section{Discussions \label{disc_section}}
The total utility achieved by DeepPR surpasses that of the other methods, since DeepPR integrates random actions with a well-behaving heuristic, RATIO to seek new experiences. This randomness prevents the proposed method from experiencing the adversarial scenarios from which RATIO suffers.

Another benefit observed is the learning speed of Deep RL. In general, it is more difficult to train a NN as a $Q$-value estimator for discrete state-action spaces than to do so for continuous spaces due to the sparser relation among neighboring points. Therefore, more episodes could be required to obtain a more accurate estimator if we adopt a sample random exploration. However, our result indicates that it is possible to speed up the learning process by integrating a well-behaving heuristic that is specific to each problem. Figure \ref{curve_diff} compares the utilities at each episode of DeepPR with the simple RANDOM exploration and the integrated exploration ($\omega_{\mathrm{RATIO}} = 0.5$). Clearly, the exploration is conducted more effectively when a heuristic is incorporated, and higher utility values are achieved in earlier episodes. 

\def\scale{0.47}
\begin{figure}[t]
\centering
\includegraphics[width=\scale\textwidth]{\figpath 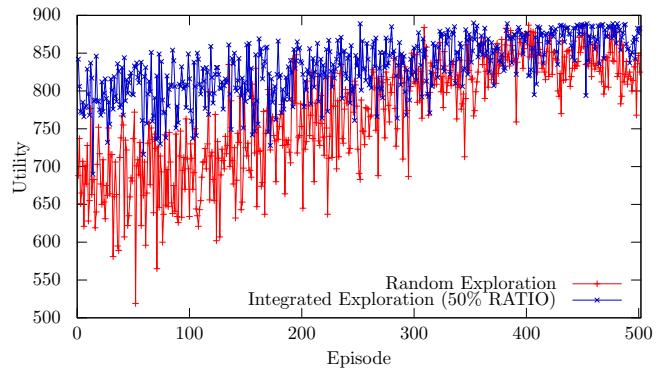}
\caption{Learning Curves with Different Exploration Methods (A GNP Random Graph): The utilization of RATIO as one of the exploration methods ($\omega_{\mathrm{RATIO}} = 0.5$) helps DeepPR to reach a better solution in earlier episodes.}
\label{curve_diff}
\end{figure}
Our results suggest the applicability of Deep RL to a wide range of optimization problems for which some heuristic algorithms are proposed. When a heuristic---which performs well for compliant cases and suffers from some critical cases---is known for the problem, it seems obvious that the addition of some degree of randomness in the action selection helps the algorithm to receive important experiential feedback, especially from the critical cases. Deep RL, in general, can extract the characteristics of such feedbacks and remembers them for future actions. Therefore, our results imply that the integration of such a simple heuristic algorithm and the Deep RL technique would provide a general methodology to design algorithms for such optimization problems, which may outperform the existing heuristic. Additionally, from the previous discussion, the integrated exploration can be more effective than general exploration methods due to the problem-specific tuning.

\section{Conclusion}
This paper discusses a progressive recovery problem of interdependent networks to maximize the total available computation utility of the networks, where a limited amount of resources arrives in a time sequence. Considering the dependency between network functions and infrastructure nodes in networks, in the problem, a node is said to be recovered when (1) the node is reachable from at least one control function and (2) the recovery resource allocation cumulatively satisfies its repairing demand, which represents the cost to fix the node itself. It is proved that the recovery problem with a general network topology always has an equivalent progressive recovery problem with a one-layered graph, which is much simpler but still NP-hard. Through a preliminary simulation result, it is also discussed that a simple heuristic algorithm called RATIO, which determines the resource allocation based on the ratio of utility and demand, can perform well when the network does not contain a specific substructure. In order to cope with the scenarios adversarial for RATIO, a deep reinforcement learning-based algorithm, DeepPR, is introduced with the exploration based on RATIO. The simulation results indicate that it achieves the near-optimal solutions in smaller real and random networks and is robust against the adversarial cases. Furthermore, it is empirically shown that the integration of RATIO and reinforcement learning improves the effectiveness of exploration of the learning. Our success in the integration suggests possible improvements of existing heuristic approaches for general optimization problems using reinforcement learning.

\bibliographystyle{ieeetr} 
\bibliography{ref}


\appendices

\section{Exponential OPT Algorithm based on Dynamic Programming \label{dpopt}}
DP-OPT calculates the optimum network-wide utility by a bottom-up dynamic programming technique, which enables us to obtain the optimum until relatively larger graphs compared to simple enumerations.


\def\rnd{\mathrm{randomInt}}
\begin{algorithm}[t]
\caption{DP-OPT$\left(G\right)$}
\label{staropt}
\begin{algorithmic}[1]
\REQUIRE A graph $G = (V, E)$, a dictionary $Z$ (size $2^{|V|}$), a dictionary $B$ (size $2^{|V|}$)\\
Note that $[V]^s$ denotes a family of $V$'s subsets of size $s$
\STATE $B[\emptyset] \gets \emptyset$
\STATE $Z[\emptyset] \gets 0$
\FOR{$s \gets 1$ \textbf{to} $|V|$} \label{for_sizeofsubset}
	\FOR{$X \in [V]^s$} \label{for_allsubset}
		\STATE $q \gets - \infty$
		\STATE $X' \gets \{v_i \in X \mid \exists v_j \in V \setminus X$ s.t. $(v_i, v_j) \in E\}$
		\FOR{$v_i \in X'$} \label{for_pickanode}
			\STATE $q' \gets u(v_i) \left(1 + \ceiling{\frac{\sum_{v_j \in X \setminus \{v_i\}} d(v_j)}{C}}\right) + $ Z$\left(X \setminus \{v_i\}\right)$
			\IF{$q < q'$}
				\STATE $q \gets q'$
				\STATE $B[X] \gets v_i$ 
			\ENDIF
		\ENDFOR
		\STATE $Z[X] \gets q$
	\ENDFOR
\ENDFOR
\label{reconst_end}
\end{algorithmic}
\end{algorithm}

\begin{lemma}
Algorithm \ref{staropt} is the optimum (exponential) algorithm to solve the progressive recovery problem in a general graph with \conditionA, where $C$ is a constant representing the amount of available resource. 
\end{lemma}
\begin{proof}
\ifjournal
When \conditionA, any recovery plan recovers at most one node at each time step. The duration that a node $v_i$ is functional is the duration that the rest of nonfunctional nodes are recovered: $\ceiling{\frac{\sum_{v_j \in X \setminus \{v_i\}} d(v_j)}{C}}$, where $X$ is a set of nonfunctional nodes. Hence, the total utility to which the recovered node $v_i$ contributes until the last step is $\ceiling{\frac{u(v_i) \cdot \sum_{v_j \in X \setminus \{v_i\}} d(v_j)}{C}}$.

The remaining problem is the same problem with $X\setminus\{v_i\}$ to recover the rest of nonfunctional nodes. The problems with smaller subsets are already solved in previous loops. Thus, the algorithm can reuse the pre-calculated results stored in $Z$. Therefore, the value of $q$, which is always the maximum for subsets of the same size, reaches the optimum when $s = |V|$.


This algorithm does not depend on any assumption on specific graph topology, since it solves the problem in a logical star graph. The set $X'$, which represents a set of nodes adjacent to any functional nodes, implicitly composes the logical star graph for each time step.
\end{proof}

\begin{lemma}
The complexity of Algorithm \ref{staropt} is $O((|V| + |E|) 2^{|V|})$.
\end{lemma}
\begin{proof}
\ifjournal
The two for-loops in line \ref{for_sizeofsubset}-\ref{for_allsubset} collectively go through all subsets in the power set of $V$. Also, in line \ref{for_pickanode}, the algorithm calculates the network-wide utility by attempting to recover each node $v_i \in X$ that is adjacent to a functional node. In the worst case, the size of $X'$ is $|V|-1$, and it requires a traversal of $E$ to check the adjacency. Thus, the computation complexity is $O((|V| + |E|) 2^{|V|})$.
\end{proof}

\fi

%
%

\newpage

\begin{IEEEbiography}[{\includegraphics[width=1in,height=1.25in,clip,keepaspectratio]{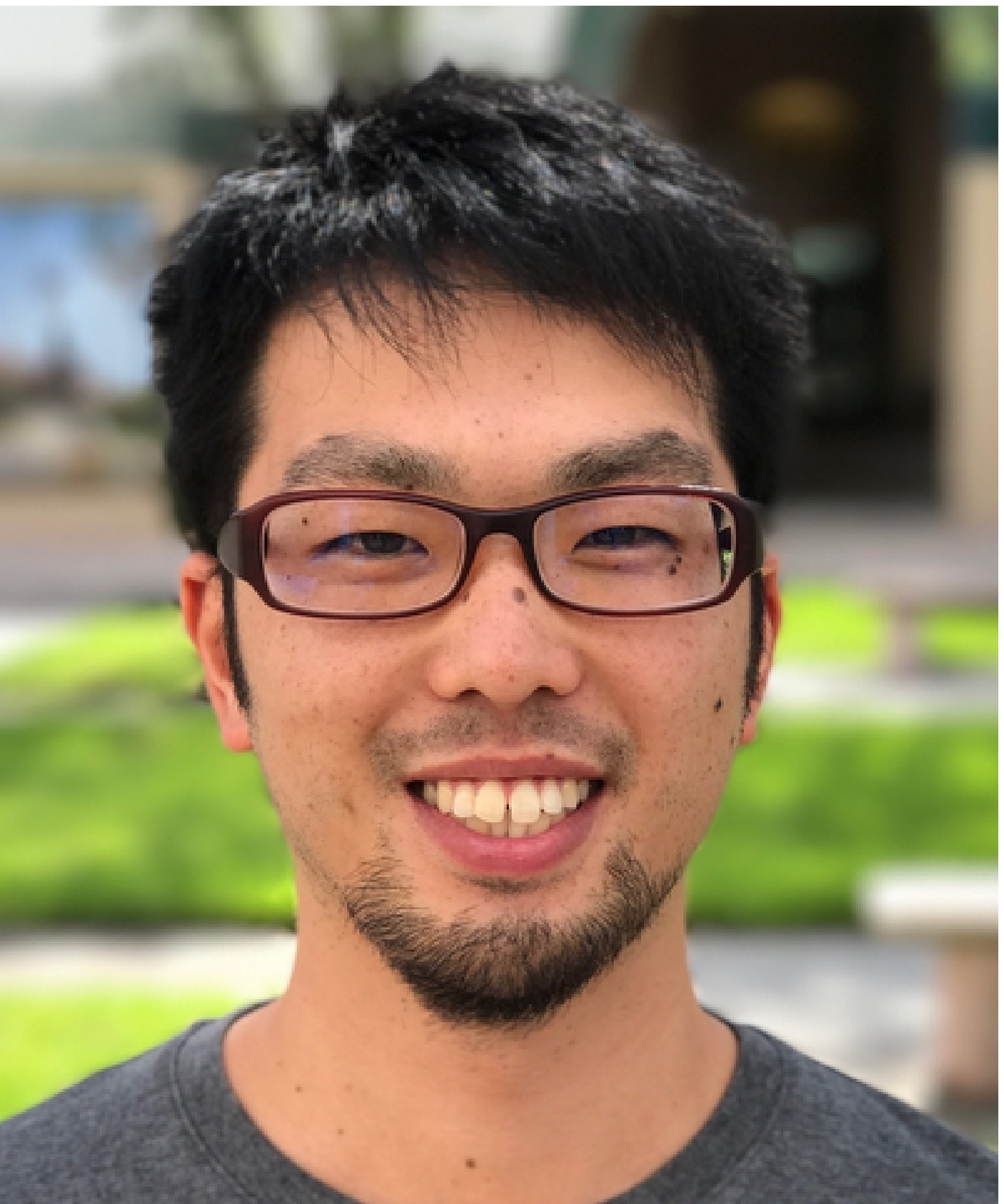}}]{Genya~Ishigaki} (GS'14) received the B.S. and M.S. degrees in Engineering from Soka University, Tokyo, Japan, in 2014 and 2016, respectively. He is currently pursuing the Ph.D. degree in Computer Science at The University of Texas at Dallas, Texas, USA. His current research interests include design and recovery problems of interdependent networks, online combinatorial optimization, and deep reinforcement learning.
\end{IEEEbiography}

\begin{IEEEbiography}[{\includegraphics[width=1in,height=1.25in,clip,keepaspectratio]{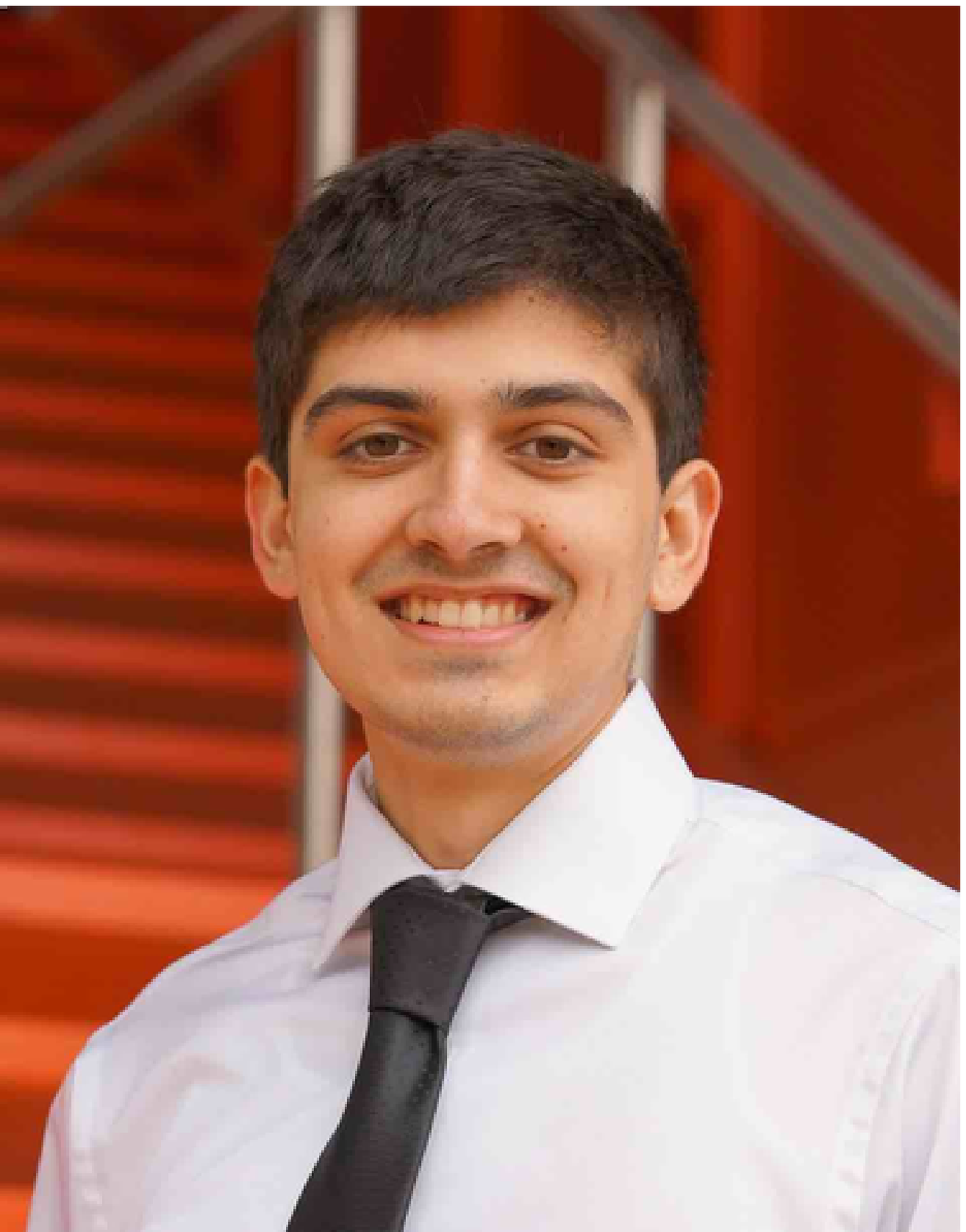}}]{Siddartha~Devic}
is currently pursuing a B.S. in Computer Science and Mathematics at The University of Texas at Dallas. He is a part of the Advanced Networks Research Laboratory, and also works with various professors on theoretical machine learning. His current research interests are graph algorithms, convex optimization, and learning theory.
\end{IEEEbiography}

\begin{IEEEbiography}[{\includegraphics[width=1in,height=1.25in,clip,keepaspectratio]{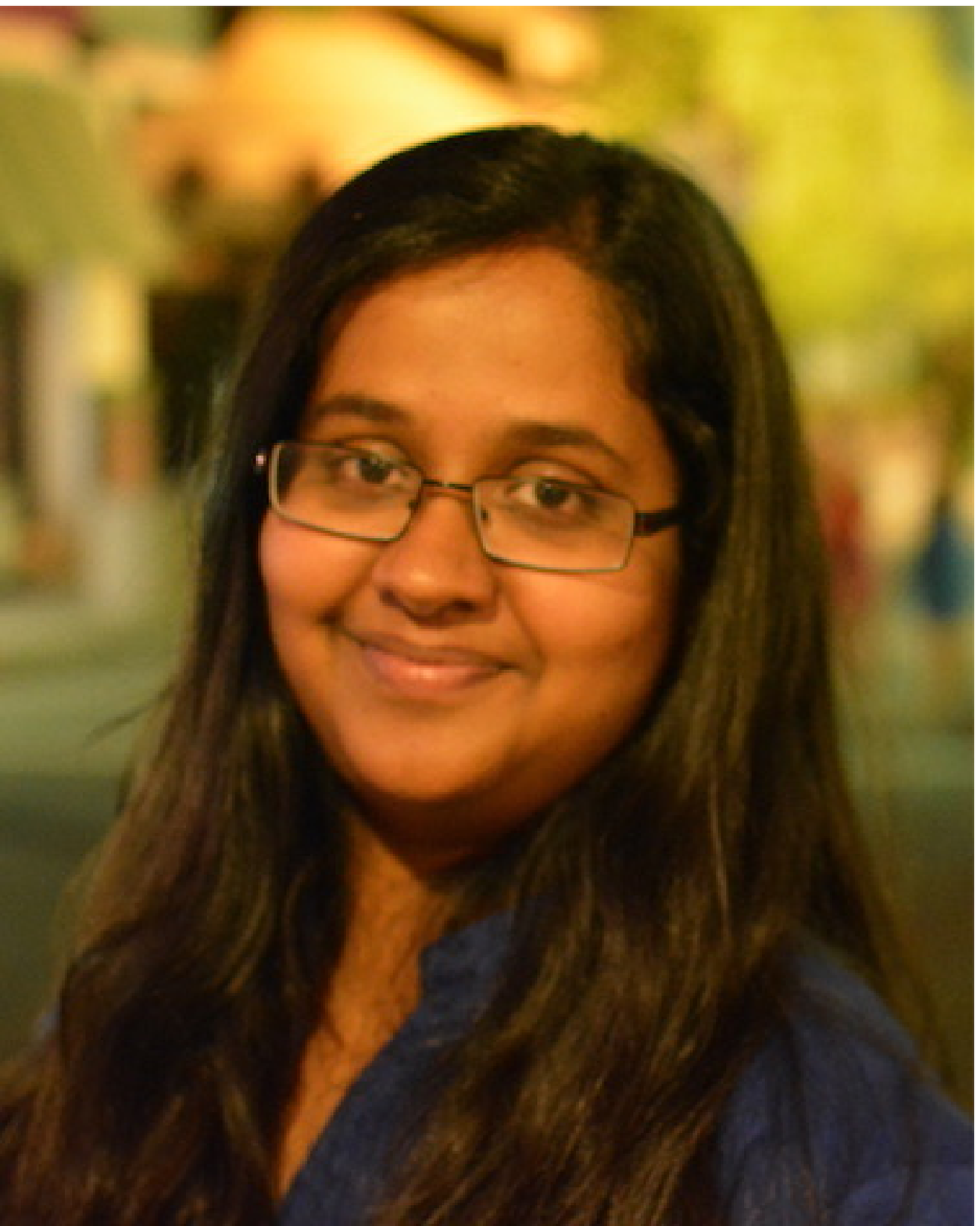}}]{Riti~Gour}
received her B.E. in Electronics and Telecommunication Engineering from C.S.V.T.U., India, in 2012, and her M.S. in Telecommunications Engineering from The University of Texas at Dallas, USA, in 2015. She is currently a PhD student in Telecommunications Engineering at UTD. Her current research is on availability of optical networks, network virtualization (slicing) and combinatorial optimization.
\end{IEEEbiography}

\begin{IEEEbiography}[{\includegraphics[width=1in,height=1.25in,clip,keepaspectratio]{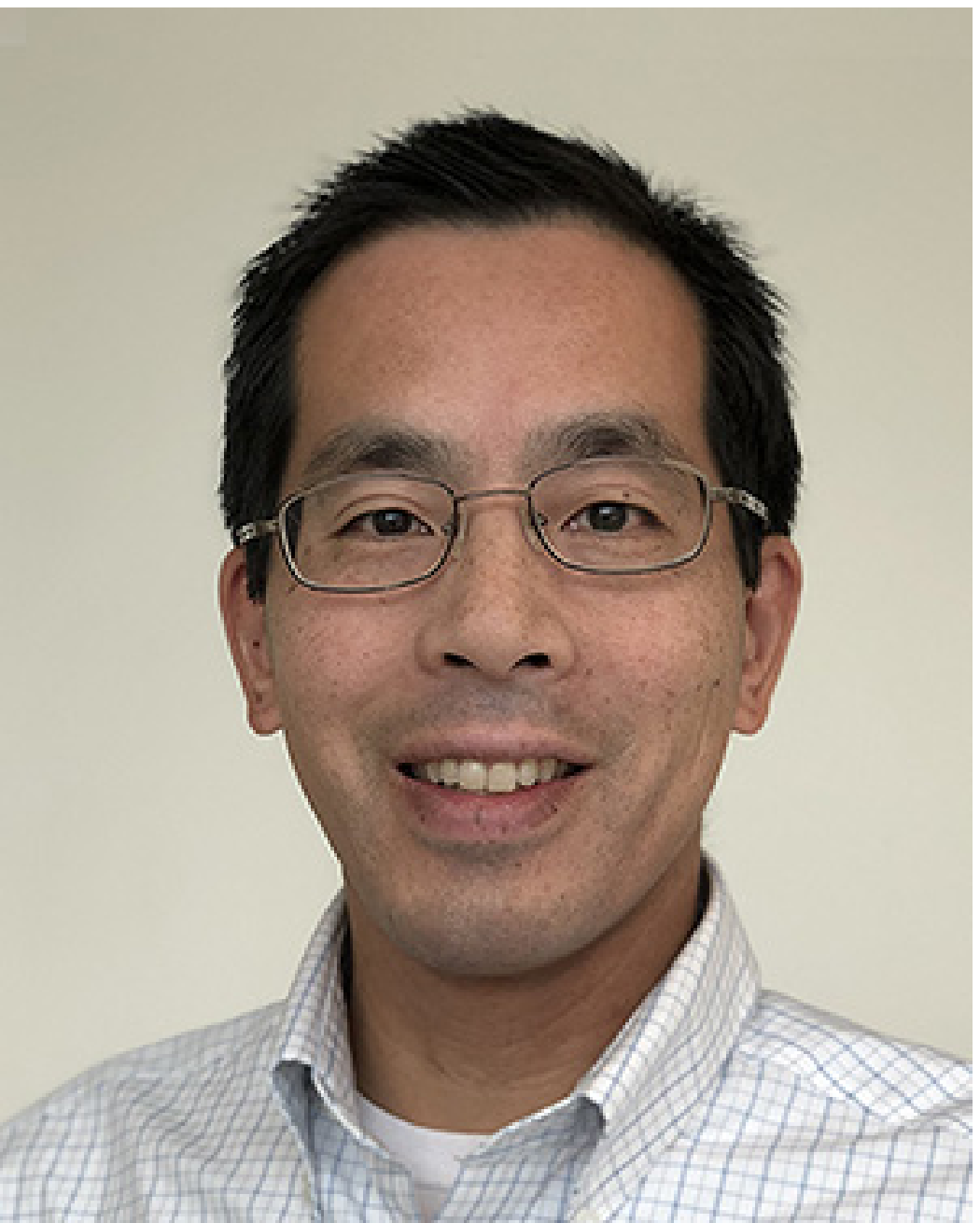}}]{Jason~P.~Jue} (M'99-SM'04) received the B.S. degree in Electrical Engineering and Computer Science from the University of California, Berkeley in 1990, the M.S. degree in Electrical Engineering from the University of California, Los Angeles in 1991, and the Ph.D. degree in Computer Engineering from the University of California, Davis in 1999. He is currently a Professor in the Department of Computer Science at the University of Texas at Dallas. His current research interests include optical networks and network survivability.
\end{IEEEbiography}

\end{document}